\pdfoutput=1
\documentclass[a4paper,UKenglish,cleveref, autoref]{lipics-v2021}

\hideLIPIcs  

\usepackage{colortbl}

\usepackage{graphics}
\usepackage[T1]{fontenc}
\usepackage{slantsc}
\usepackage{lmodern}
\AtBeginDocument{%
	\DeclareFontShape{T1}{lmr}{m}{scit}{<->ssub*lmr/m/scsl}{}%
}

\usepackage[ruled]{algorithm2e}

\usepackage{tikz}
\usetikzlibrary{shapes}
\usetikzlibrary{decorations.markings,shadows, shapes}
\usetikzlibrary{patterns, positioning,shapes.misc,matrix, arrows.meta, calc} %

\newcommand{\Cross}{\mathbin{\tikz [x=1.4ex,y=1.4ex,line width=.2ex] \draw (0,0) -- (1,1) (0,1) -- (1,0);}}%

\title{On the Resilience of Fast Failover Routing\linebreak  Against Dynamic Link Failures} 

\titlerunning{On the Resilience of Fast Failover Routing Against Dynamic Link Failures} 

\author{Wenkai Dai}{Faculty of Computer Science and UniVie Doctoral School Computer Science DoCS, University of Vienna, Austria}{wenkai.dai@univie.ac.at}{https://orcid.org/0000-0002-2153-4250}{}

\author{Klaus-Tycho Foerster}{Department of Computer Science, Technical University of Dortmund, Germany}{klaus-tycho.foerster@tu-dortmund.de}{https://orcid.org/0000-0003-4635-4480}{}

\author{Stefan Schmid}{TU Berlin, Germany\\Faculty of Computer Science, University of Vienna, Austria}{stefan.schmid@tu-berlin.de}{https://orcid.org/0000-0002-7798-1711}{}

\authorrunning{W. Dai, K.-T. Foerster, and S. Schmid}

\Copyright{W. Dai, K.-T. Foerster, and S. Schmid}



\category{} 

\relatedversion{} 

\funding{Supported by the Vienna Science and Technology Fund (WWTF), grant 10.47379/ICT19045 (WHATIF), 2020-2024.}

\nolinenumbers 

\EventLongTitle{}
\EventShortTitle{}
\EventAcronym{}
\EventYear{2024}
\EventDate{}
\EventLocation{}
\EventLogo{}
\SeriesVolume{1}
\ArticleNo{1}

\begin{document}

\maketitle

\begin{abstract}
Modern communication networks feature local fast failover mechanisms in the  data plane, to swiftly respond to link failures with pre-installed rerouting rules. This paper explores resilient routing  meant to tolerate  $\leq k$ simultaneous link~failures, ensuring packet delivery, provided that the source and destination remain connected.  While past theoretical works studied  failover routing under static link failures, i.e., links which permanently and simultaneously fail,  real-world networks often face link flapping—dynamic down states caused by, e.g.,  numerous short-lived software-related faults. Thus, in this initial work, we  re-investigate the resilience of failover routing against link flapping, by categorizing link failures into static, semi-dynamic (removing the assumption that links fail simultaneously), and dynamic (removing the assumption that links fail permanently) types, shedding light on the capabilities and limitations  of failover routing under these scenarios.

We show that $k$-edge-connected graphs exhibit $(k-1)$\nobreakdash-resilient routing against dynamic failures for $k\leq 5$. We further show that this result extends to arbitrary $k$ if it is possible to rewrite $\log k$ bits in the packet header.
 Rewriting $3$ bits suffices to cope with $k$ semi-dynamic failures. However, on general graphs, tolerating $2$ dynamic failures becomes impossible without bit-rewriting. Even by rewriting $\log k$ bits, resilient routing cannot resolve $k$ dynamic~failures, demonstrating the limitation of local fast rerouting.
\end{abstract}

\section{Introduction and Related Work}
Communication networks are a critical infrastructure of our digital society.  To fulfill their rigorous reliability requirements, networks must possess the capability to promptly handle link failures, which are inevitable and likely to grow in frequency as the scale of networks is expanding~\cite{gill2011understanding}. Studies have demonstrated that even brief disruptions in connectivity can lead to significant degradation in service quality~\cite{dctcp, d2tcp, zats2012detail}. When faced with failures,  routing protocols  like OSPF~\cite{rfc2328} and IS-IS~\cite{isis}  recalculate routing tables, however, such reactions in the control plane are slow; in fact, too slow for many latency-sensitive applications~\cite{dctcp, d2tcp, zats2012detail,gill2011understanding,Francois05b}.

Modern networks hence incorporate local fast failover  mechanisms in the data plane to respond faster to failures. These mechanisms enable rapid rerouting of packets along preinstalled alternative paths, eliminating the necessity for global route recomputation and resulting in recovery times that can be orders of magnitude faster~\cite{liu2013ensuring,feigenbaum2012brief}. For instance, many networks employ \emph{IP Fast Reroute}~\cite{ipfrr-survey-1, ipfrr-survey-2, ipfrr-evolution} or \emph{MPLS Fast Reroute}~\cite{mplsfrr} to address failures in the data plane. In  Software-defined Networking (SDNs), fast rerouting (FRR) functionality is provided through \emph{OpenFlow fast-failover groups}~\cite{offrr}.

The quest for efficient failover mechanisms  within the data plane involves a significant algorithmic challenge. The primary dilemma revolves around the strategic establishment of static failover rules that can swiftly reroute flows to maintain reachability at the routing level,  in the face of failures, while these routing rules must solely depend on local failure information, without information about potential downstream failures. At the core of this challenge lies a fundamental question:
\begin{itemize}
\item Can we devise a failover routing system capable of tolerating any $k$ simultaneous link failures as long as the underlying topology remains connected?
\end{itemize}

Recent years have seen a significant surge of interest in the development of resilient failover mechanisms~\cite{chiesa2020fast}. Randomization techniques~\cite{icalp16,DBLP:journals/ton/ChiesaSABKNS21},  such as employing a random walk, prove effective in handling failures in a graph as long as it is connected, akin to graph exploration, but standard routers and switches lack efficient support for randomized forwarding~\cite{DBLP:conf/infocom/FoersterP0T19,DBLP:journals/tpds/LeungLY07}, and randomized forwarding can lead to a high number of packet reorderings~\cite{DBLP:conf/infocom/FoersterP0T19}, resulting in a substantial decrease in throughput~\cite{DBLP:conf/infocom/FoersterP0T19}. Similar challenges apply to packet duplication algorithms, like flooding, which also impose a significant additional load on the network.

As a result, standard routers and switches predominantly utilize deterministic routing functions, which base their forwarding decisions on active links connected to a switching device, the incoming port of a packet, and the destination. We will refer to this straightforward type of deterministic routing function as \emph{basic} routing. In pursuit of this basic forwarding approach, Feigenbaum et al.~\cite{feigenbaum2012brief} introduced  an approach based on directed acyclic graphs (DAGs), that  can always guide failover routing against one failure. However, resilience to arbitrarily many link failures,  hereinafter referred to as \emph{\textbf{perfect resilience}}, becomes impossible even on a graph of eight nodes~\cite{DBLP:conf/wdag/FoersterHP0T20}.
Meanwhile, Chiesa et al.~\cite{DBLP:journals/ton/ChiesaNMGMSS17} demonstrated the impossibility of achieving $2$-resilience in general when disregarding the source information. Consequently, many subsequent solutions offer heuristic guarantees, rely on packet header rewriting, or necessitate densely connected networks~\cite{icalp16,chiesa2020fast,foerster2020feasibility,algocloud2024,keep-fwd}.  In scenarios where packet header rewriting is possible, a router can interpret and rewrite the reserved bits in the header of an incoming packet to influence the forwarding decisions made by itself and also subsequent routers handling the packet.

Under dense connectivity assumptions, Chiesa et al.~\cite{DBLP:journals/ton/ChiesaNMGMSS17} show that the $(k-1)$\nobreakdash-resilient routing functions in  $k$-edge-connected graphs can be efficiently computed for $k\leq 5$, which is called  \emph{\textbf{ideal resilience}}  (in this case for $k\leq5$). But, the question of the  ideal resilience of $> 5$ still remains unresolved so far. Additionally, leveraging packet header modification, Chiesa et al.~\cite{DBLP:journals/ton/ChiesaNMGMSS17} have devised two failover algorithms capable of achieving $(k-1)$\nobreakdash-resilience in arbitrarily large $k$-edge-connected networks. This is accomplished by rewriting $\log k$ and $3$ bits in packet headers, respectively.  Note that, in the following, unless specified otherwise, we assume that failover routing functions do not rewrite bits  in packet headers.

Nevertheless, numerous real-world networks exhibit relatively sparse connectivity, with certain dense regions, such as enterprise networks, ISP networks or wide-area networks~\cite{DBLP:conf/nca/MontgolfierSV11, DBLP:conf/infocom/FoersterKP0T21}. Consequently, the development of fast failover routing algorithms for arbitrary topologies  becomes highly desirable in such scenarios.
%
Dai et al.~\cite{DBLP:conf/spaa/DaiF023} demonstrate that a $2$-resilient failover routing,  depending on  source and destination, can be efficiently computed in general graphs but achieving  $\ge3$-resilience is not possible.  For clarity, we will refer to the routing function  utilizing also source information in addition to destination information,  as \emph{source-matched}~routing.

Up until now, existing theoretical studies on failover routing have mostly assumed \emph{simultaneous} failures for \emph{fail-stop} links, where  links fail simultaneously, and links stay failed forever, once they fail.
However, in practice, links may not fail simultaneously, but more likely to fail over time, and the failed links may also be restored. 
For example, in wide area and enterprise networks, a phenomenon known as \emph{link flapping}, where communication links alternate between up and down states, is frequently observed under external routing protocols, e.g., in OSPF~\cite{10.1145/637201.637236,10.5555/850929.851939} and IS-IS~\cite{markopoulou2008characterization-malformed,turner2010california}.

More recently, Gill et al.~\cite{gill2011understanding} reveal that link failures in data centers can be categorized into long-lived and sporadic short-lived failures respectively, possibly implicated by connection errors, hardware issues, and software errors. They further comment that software errors are more inclined for short-lived failures~\cite{gill2011understanding}. It is important to note that link flapping transforms the underlying network topology into a dynamic graph, which could significantly impact these existing failover algorithms, as they may heavily depend on locally static structures for~orientation.

Hence, the focus of this paper is to explore and establish provable and deterministic worst-case resilience guarantees under the influence of link flapping. Specifically, we characterize link failures  by three types: \emph{static, semi-dynamic,} and \emph{dynamic}. For dynamic failures,  unstable links alternate between up and down arbitrarily, whereas in semi-dynamic failures,  unstable links are fail-stop  but  do not have to fail simultaneously and may fail \emph{during} the packet traversal. In static failures, unstable links are fail-stop  and fail simultaneously.
Based on this classification, we note that static failures constitute a subset of semi-dynamic failures, and semi-dynamic failures are encompassed within a subset of dynamic failures. Consequently, a resilient routing algorithm designed for a specific type of failure can be readily utilized to address failures within its subset. Furthermore, any impossibility result for a particular failure type also applies to failures in the supersets, but not vice versa.

In this paper, our consideration is limited to basic routing functions and their variations, including packet header rewriting and source-matching. As such, our work is closely related to the studies conducted by Chiesa et al.~\cite{DBLP:journals/ton/ChiesaNMGMSS17} and Dai et al.~\cite{DBLP:conf/spaa/DaiF023}. For a comprehensive overview of the directly related findings, please refer to Table~\ref{table:overview}. Going forward, we aim to re-evaluate whether the conclusions drawn by~\cite{DBLP:journals/ton/ChiesaNMGMSS17} and~\cite{DBLP:conf/spaa/DaiF023}, which were established primarily for static failures, can hold true for dynamic and semi-dynamic failures as well.

\begin{table*}[t]
	\definecolor{lgray}{rgb}{0.95,0.95,0.95}
	\caption{Summary of related previous results and our new contributions, where previous results are presented in white-gray rows, and our new findings related to the ideal-resilience and the perfect-resilience are cast into green and blue rows respectively.}\label{table:overview}
	\centering
		\resizebox{\textwidth}{!}{
			\renewcommand{\arraystretch}{1.1}
			\begin{tabular}{cccccc}
				\textbf{Failure} & \textbf{Rewriting}&\multicolumn{2}{c}{\textbf{Routing information}} & \textbf{Graph is} &  	\textbf{Resilience results for}\\
				\textbf{Type}& 	\textbf{Bits \#} &  \textbf{\textit{Per-source}}& \textbf{\textit{Per-destination}}&   \textbf{$k$-edge-connected} &\textbf{static \& deterministic routing}  \\
				\hline
				\hline

				static & no &  &$\Cross$  &$\Cross$ &$(k-1)$-resilience possible for $k\leq 5$, \textbf{open} for $k\ge 6$~\cite{DBLP:journals/ton/ChiesaNMGMSS17}\\

				\rowcolor{lgray}	static & no &  &$\Cross$  &$\Cross$ &$\left( \lfloor k/2 \rfloor\right) $-resilience possible for any $k$~\cite{DBLP:journals/ton/ChiesaNMGMSS17}\\
				static & no & $\Cross$  & $\Cross$   &$\Cross$  &  	$(k-1)$-resilience possible for any $k$~\cite{DBLP:conf/infocom/FoersterP0T19}\\
				\rowcolor{lgray}			 static & $3$ &   & $\Cross$  &$\Cross$ & 	$(k-1)$-resilience~\cite{DBLP:journals/ton/ChiesaNMGMSS17}\\
				static  & no &  &$\Cross$ &  & $1$-resilience possible~\cite{feigenbaum2012brief}, $\ge 2$-resilience imposs.~\cite{DBLP:journals/ton/ChiesaNMGMSS17}\\
				\rowcolor{lgray}  	static & no & $\Cross$  & $\Cross$  & & arbitrary $k$-resilience impossible~\cite{foerster2020feasibility}\\

				static & no & $\Cross$  & $\Cross$  & & $2$-resilience possible, $3$-resilience impossible~\cite{DBLP:conf/spaa/DaiF023}\\
				\hline
				\rowcolor{green!15}	dynamic  & no & & $\Cross$ &$\Cross$ &  $\left( k-1\right) $-resilience possible for $k\leq 5$~[Thm.~\ref{thm: 3-resilience}--\ref{thm: 5-resilience}]\\
				\rowcolor{green!7}	dynamic  & no & & $\Cross$ &$\Cross$&  $\left( \lfloor k/2 \rfloor\right) $-resilience possible for any $k$~[Thm.~\ref{thm: k/2-resilience}]\\
				\rowcolor{green!15}	dynamic  & $\log k$ &  &$\Cross$  &$\Cross$ & $\left( k-1\right) $-resilience  possible for any $k$~[Thm.~\ref{thm: log-k-bits-reilsiency}] \\
				\rowcolor{green!7}	semi-dynamic  & $3$ & & $\Cross$   &$\Cross$ & $\left( k-1\right) $-resilience possible for any $k$~[Thm.~\ref{thm: semi-dynamic-3bits}] \\
				\rowcolor{green!15}	dynamic  & $3$ & & $\Cross$   &$\Cross$ &  HDR-$3$-\textsc{Bits} in~{\cite[Algorithm~2]{DBLP:journals/ton/ChiesaNMGMSS17}} inapplicable [Thm.~\ref{thm: counter-example-3bits}]\\
				\rowcolor{green!7}	dynamic  &no &  & $\Cross$   &$\Cross$ &  $\left( k-1\right) $-resilience imposs. for $k\ge 2$ (link-circular)  [Thm.~\ref{thm: no-1-resilience}]\\

				\rowcolor{blue!15}	 dynamic &no & & $\Cross$   &  &  $1$-resilience possible, $\ge 2$-resilience imposs.\~[Thm.~\ref{thm: 1-resilience-always}, \ref{thm: 2-resilient-with-1bit}]\\
				\rowcolor{blue!7}	 dynamic &no & $\Cross$ & $\Cross$   &  &   $\ge 2$-resilience is impossible [Thm.~\ref{thm: 2-resilient-with-1bit}]\\
				\rowcolor{blue!15} dynamic	& $\log k $ & $\Cross$  & $\Cross$ &  &  	arbitrary $k$-resilience impossible [Thm.~\ref{thm: no-perfect-resilience}] \\
				\rowcolor{blue!7} semi-dynamic	& no & $\Cross$  & $\Cross$ &  &  $2$-resilient	 algo. (static failures) in~\cite{DBLP:conf/spaa/DaiF023} inapplicable  [Thm.~\ref{thm: counter-Dai-SPAA'23}] \\
				\hline

		\end{tabular} }
		\label{tab1}

\end{table*}
\subsection{Contributions}

This paper  initiates the study of the achievable resilience of fast rerouting mechanisms against more dynamic and non-simultaneous link failures. To this end, we chart a landscape of rerouting mechanisms under static, semi-dynamic, and dynamic link failures. We provide a summary of our results in Table~\ref{table:overview}.

We demonstrate that, in  $k$-edge-connected graphs with $k\leq 5$, achieving $(k-1)$-resilience under dynamic failures is possible without rewriting bits in packet header or employing source-matching in routing functions. This ideal $(k-1)$-resilience under dynamic failures extends to any $k$ if routing functions can rewrite $\log k$ bits in packet header. However, the HDR-$3$-\textsc{Bits} Algorithm by Chiesa et al.~\cite{DBLP:journals/ton/ChiesaNMGMSS17}, offering an arbitrary $(k-1)$\nobreakdash-resilience by rewriting $3$ bits for static failures, is no longer viable for dynamic failures but remains productive for semi-dynamic failures.

Contrarily, for general graphs, the $2$-resilient source-matched routing algorithm proposed by Dai et al.~\cite{DBLP:conf/spaa/DaiF023} for static failures,  does not extend to  semi-dynamic failures, and achieving $2$-resilience through matching source against dynamic failures is not possible in general.
Furthermore, achieving the source-matched perfect resilience (i.e., $k$-resilience for any $k$) on general graphs is impossible even with the ability to rewrite $\log k$ bits. Although the $1$-resilience for dynamic failures without using  bits is always  feasible, the $1$-resilience in a $2$-edge-connected graph becomes impossible if all nodes must employ \emph{link-circular routing} functions.
\subsection{Organization}

The remainder of this paper is organized as follows. We introduce our formal model in~\S\ref{sec:Preliminaries} and then present our algorithmic results for ideal-resilience and perfect-resilience against various dynamic failures in \S\ref{sec: ideal} and \S\ref{sec: perfect} respectively. We conclude in \S\ref{sec:conclusion} by discussing some open~questions.

\section{Preliminaries}
\label{sec:Preliminaries}
We represent a given network as \emph{an undirected (multigraph)} $G = (V,E)$, where each router in the network  corresponds to a node in $V$ and each  bi-directed link between two routers  $u$ and $v$ is modeled  by an \emph{undirected edge} $\{u,v\}\in E$.

For a set of edges $E' \subset E$, let $G \setminus E'$  denote a subgraph $\left(V, E\setminus E' \right) $ of $G$, and for  a set of  nodes $V' \subset V$,  let $G \setminus V'$  be a subgraph of $G$ obtained by  removing $V'$ and all incident edges on $V'$ in $G$. For a graph $ G'\subseteq G$  and a node $v\in V\left( G'\right)  $, let $N_{G'}\left( v\right) $, $E_{G'}\left( v\right) $, $\Delta_{G'}\left(v \right) $ denote the \emph{neighbors} (excluding $v$), \emph{incident edges}, and \emph{the degree} of  $v$  in $G'$ respectively, where $G'$ can be omitted when the context is clear. For an  undirected edge $\{u, v\}\in E$, let $\left( u,v\right) $ (resp., $\left(v, u\right) $) denote a \emph{directed edge (arc)} from $u$ to $v$ (resp., from $v$ to $u$). 

\subparagraph*{Static, Semi-Dynamic, and Dynamic Failures.}
Let $F\subseteq E$ denote a set of \emph{unstable links (failures)}   in $G$, where each $e\in F$ can fail in transferring packets in both directions when its link state is \emph{down (failed)}. An unstable link is called \emph{fail-stop} if its down state becomes permanent once it is down. The set of (unstable links) failures $F\subseteq E$ can be classified into three types: \emph{semi-dynamic} if all links in $F$ are fail-stop,  \emph{static}  if all links  in $F$ are fail-stop and must fail simultaneously,  otherwise \emph{dynamic}, if  $\exists e\in F$ can alternate   between \emph{up} and \emph{down} states arbitrarily and can fail over time. Based on this classification,  static failures are special cases of semi-dynamic failures, and semi-dynamic failures are a subset of dynamic failures.

Consequently, a resilient routing algorithm designed for a specific failure type can be directly applied to handle   failures in its subset, and any impossibility result for a particular failure type also holds true for  failures in its superset, but not vice versa.

\subparagraph*{Failover Routing.} For a basic (resp., source-matched) failover routing,  each  node $v\in V$ stores a \emph{predefined and static} \emph{forwarding  (interchangeably, routing) function} to \emph{deterministically} decide an \emph{outgoing link (out-port)} for each incoming packet solely relying on the local information  at $v$,~i.e.,
\begin{itemize}
	\item   the destination $t$ (resp., the source $s$ and the destination $t$) of the incoming packet,
	\item the incoming link (\emph{in-port}) of the packet at node $v$,
	\item and the set of non-failed (active) links incident on $v$.
\end{itemize}

Specifically,  given a graph $G$ and a destination $t\in V$ (resp.,  \emph{source-destination pair}  $(s,t)$), a  \emph{forwarding function} for a destination $t$ (resp.,  \emph{source-destination pair}  $(s,t)$) at a node $v\in V$ is defined as $\pi^{t}_{G,v} \colon N_G(v) \cup \{ \bot \} \times 2^{E_G(v)}  \mapsto E_G(v)$ (resp., $\pi^{\left(s,t \right) }_{G,v} \colon N_G(v) \cup \{ \bot \} \times 2^{E_G(v)}  \mapsto E_G(v)$), where $\bot$ represents sending a packet originated at $v$ (these functions can be extended appropriately  for multigraphs). Unless otherwise stated, we  will implicitly consider forwarding functions  without matching the source $s$.

When $G$ and the source-destination pair  $(s,t)$ are clear, $\pi^{(s,t)}_{G,v}\left( u, E_{G\setminus F}(v)\right)$ can be abbreviated as $\pi_{v}\left(u, E_{G\setminus F}(v) \right)  $, where $u\in N_{G\setminus F}(v)\cup \{ \bot \}$.
Let $F_v\subseteq F$ denote the failures that  incident on a node $v\in V$.
With a slight abuse of notation,  the forwarding function $\pi^{(s,t)}_{G,v}\left( u, E_{G\setminus F}(v)\right)$ can be also denoted by the form $\pi^{(s,t)}_{G,v}\left( u, F_v\right)$ since $E_{G\setminus F}(v)= E_G\left(v \right) \setminus F_v$.
Especially, when $v$ does not lose any link under $F$, i.e., $E_{G\setminus F}\left( v\right) =E_G\left(v\right) $, its routing function is simplified as $\pi_v\left(u\right)$.
The collection of routing functions:     $\Pi^{(s,t)} =\bigcup_{v\in V\setminus\{t\}} \left( \pi_v^{\left( s,t\right) }\right) $   is called a \emph{routing  scheme for $(s,t)$.} Similar abbreviations and terminology can be applied to $\Pi^{t} =\bigcup_{v\in V\setminus\{t\}} \left( \pi_v^{t }\right) $ without a source $s$.

\subparagraph*{Packet Header-Rewriting Routing.} To augment basic routing functions, the packet header-rewriting  routing protocol  can reserve various rewritable bits in specific positions within each packet's header, s.t., the header-rewriting routing function at each node $v\in V$ can interpret the information conveyed by these bits in the  header of an incoming packet as an additional parameter for its forwarding decision and can also modify these bits in the packet's header to influence the forwarding decisions made by subsequent routers handling the packet. 

The allotment of rewritable bits reserved for failover routing remains notably constrained, owing to concurrent demands for bit modification in critical functions such as TTL, checksums, and QoS, while the \emph{complexity of bit-rewriting} has a notable impact on processing overhead, latency, and the risk of packet loss, ultimately affecting the efficiency of transmission.

After introducing routing functions,  in Definition~\ref{def:1}, we  formally define the core problem  studied in this paper.
\begin{definition}[$k$-Resilient Failover Routing Problem]
	Given a graph $G=(V,E)$, the \emph{$k$\nobreakdash-resilient  failover routing problem} is to compute a \emph{$k$\nobreakdash-resilient} routing scheme for a destination $t$ (resp., a source-destination pair $(s,t)$) in $G$. A forwarding scheme for $t$ (resp., $(s,t)$) is called \emph{$k$\nobreakdash-resilient},  if this scheme can route a packet originated at a node $s\in V$ to its destination $t\in V$ as long as $s-t$ remains connected in~$G\setminus F$ for a set of  (static/semi-dynamic/dynamic) link failures $F\subseteq E$ of $|F|\leq k$,  where $G\setminus F$ denotes the subgraph when $F$ fails simultaneously.
	\label{def:1}
\end{definition}

We will focus on computing  a $k$\nobreakdash-resilient routing scheme for a given destination $t$ (resp., a given  pair of source-destination $\left(s,t \right)$), as our algorithm  for $t$ (resp.,  $\left( s,t\right) $) can be applied to any node $u\in V$ (resp., two nodes $u,v\in V$).

It is worth noting that since resilient routing often involves packets retracing their paths in reverse directions, dynamic (or semi-dynamic) failures can cause inconsistencies in the input (active links) for deterministic routing functions, thereby increasing the likelihood of forwarding loops.

\subparagraph*{Non-Trap Assumption.}
In Definition~\ref{def:1}, the subgraph $G\setminus F$ may consist of multiple connected components. We assume that dynamic (or semi-dynamic) failures in $F$ \emph{cannot} lead routing functions to direct a packet across different connected components within $G\setminus F$.

\subparagraph*{Ideal Resilience and Perfect Resilience.}
In $k$\nobreakdash-edge-connected graphs, $\left(k-1 \right) $-resilience is also called  \emph{ideal resilience}. In  general graphs, $k$-resilience for an arbitrarily large $k$ is also called \emph{perfect resilience}.

\subparagraph*{Dead-Ends, Loops, and Circular Routing.}
Next, we introduce some commonly-used concepts in failover routing.
We say that a node $v$ \emph{bounces back} a packet $p$, if $v$ sends the incoming packet $p$ back through its incoming port.
Given a graph $G'\subseteq G$, if a node $v\in V\left( G'\right) $ has only one neighbor, i.e., $\Delta_{G'} \left( v\right)=1$, then we call $v$ a \emph{dead-end},  any forwarding function at $v$ must bounce back packets received from its unique neighbor, otherwise packets must be \emph{stuck} after arriving at $v$.
A \emph{forwarding loop}  arises after a packet traverses the same direction of an undirected link for the second time.
Both directions of an undirected link can be traversed once without generating a loop. A forwarding loop  appearing  on static failures can also occur for dynamic (resp., semi-dynamic) failures.
A packet cannot be routed from $u$ to $v$ anymore, if its  routes contains a forwarding loop or if the packet is stuck at a node. 

A forwarding function at a node $v\in  V$ is called \emph{link-circular} if    $v$ routes packets based on an \emph{ordered circular sequence} $\left<u_1,\ldots, u_\ell\right>$ of the \emph{neighbors} $\{u_1,\ldots, u_\ell\}$ of $v$ as follows: a packet $p$ received from a node $u_i$ is forwarded to $u_{i+1}$; if the link $\{v, u_{i+1}\}$ is failed, then $p$ is forwarded to $u_{i+2}$ and so on, with $u_1$ following $u_\ell$~\cite{DBLP:journals/ton/ChiesaNMGMSS17}. Obviously, for link-circular forwarding functions, bouncing back is only allowed on dead-ends.

\subparagraph*{Further Notations and Graph Theory Concepts.}
In the following, we also state some related concepts in graph theory and notations that we will use.
A path $P$ from $u\in V$ to $v\in V$ in $G$ is called an $u-v$ path in $G$. Two paths are \emph{edge-disjoint} if they do not have any joint edge, but they may have common (joint) nodes.

In this paper, our focus frequently revolves around the \emph{edge-connectivity}, henceforth simply denoted by \emph{connectivity}. In  a graph $G=(V,E)$,  two nodes $u\in V$ and $v\in V$ are \emph{$k$\nobreakdash-connected}   (interchangeably, $u-v$ is \emph{$k$\nobreakdash-connected})  if there are $k$ \emph{edge-disjoint $u-v$ paths}  in $G$, and  $G$ is termed \emph{$k$\nobreakdash-connected} if  any two nodes in $V$ are  \emph{$k$\nobreakdash-connected}.
Given $V'\subseteq V$, we use $G\left[ V'\right] $ to denote an \emph{induced subgraph} of $G$  on nodes $V'$, where an edge $\{u,v\}\in E$ is also  contained in the graph $G\left[ V'\right] $  iff $u,v\in V'$.

\section{Ideal Resilience Against Dynamic Failures}\label{sec: ideal}
In this section, we focus on the ideal resilience, which seeks for a $\left( k-1\right) $\nobreakdash-resilient routing  in a $k$-connected graph against dynamic failures. We investigate this problem along two dimensions: without  or with rewriting bits in packet~headers.

\subsection{Background on Ideal Resilience against Static Failures}\label{subsec:recap}
In the following, we first recap the routing techniques proposed by Chiesa et al.~\cite{DBLP:journals/ton/ChiesaNMGMSS17} to achieve the ideal resilience against static failures and show how the established results for static failures can be effectively adapted for dynamic failures.
A detailed discussion is deferred to  Appendix~\ref{sec: techniques}.

In a $k$\nobreakdash-connected graph $G=\left( V,E\right) $, \emph{a set of $k$ arc-disjoint arborescences (directed spanning trees) $\mathcal{T}=\left\lbrace T_1, \ldots, T_{k}\right\rbrace $  of $G$, rooted at  $t\in V$}, can be computed efficiently~\cite{TARJAN197451}, both in theory (in $O\left( \left| E\right| k\log n + nk^4\log^2 n\right)$~\cite{10.5555/1347082.1347132}) and in practice~\cite{DBLP:conf/wea/GeorgiadisKMN22}.\footnote{We note that the runtime of the preprocessing is immaterial in our context, as the preprocessing is just responsible for the routing table computation and does not influence routing performance itself.} Chiesa et al.~\cite{DBLP:journals/ton/ChiesaNMGMSS17} leverage a fixed order $\left\langle \mathcal{T} \right\rangle $ of $k$ arc-disjoint arborescences  to define their \emph{circular-arborescence routing} (Definition~\ref{def: circular-arborescence}). In this routing scheme, termed the \emph{canonical mode}, packets are routed on each $T \in \langle \mathcal{T} \rangle$ by following the directed path of $T$ to $t \in V$. Upon encountering a failure  $\left(u,v \right)\in E\left(T \right)  $ on $T$, the canonical mode is reapplied on the next arborescence of $T$ in $\langle \mathcal{T} \rangle$, starting at $u\in V$. However, the circular-arborescence routing is only effective for $k \leq 4$.
For general $k$ and any $k-1$ static failures $F\subset E$, Chiesa et al.~\cite{DBLP:journals/ton/ChiesaNMGMSS17} demonstrate that, there exists a \emph{good arborescence} $T$ in $\mathcal{T}$, s.t.,   encountering any   failure $\left( u, v\right) \in E\left( T\right) $ results in a phenomenon called \emph{well-bouncing}, i.e.,  \emph{bouncing} packets from $T$ to an  arborescence $T'\in \mathcal{T}$ with $\left(v, u\right) \in E\left( T'\right)$, to resume  canonical mode on $T'$,  starting at $u\in V$, will lead to uninterrupted arrival at $t$ along $T'$.

Due to the existence of a good arborescence, Chiesa et al.~\cite{DBLP:journals/ton/ChiesaNMGMSS17} develop two packet-header-rewriting routing algorithms to incorporate the good-arborescence checking procedures into  the circular-arborescence routing over  $\langle \mathcal{T} \rangle$.

In Theorem~\ref{thm: good-arborescence} (Appendix~\ref{sec: techniques}), we  further show that a good arborescence in $\mathcal{T}$ still exists for dynamic failures $F$.
\subsection{Ideal Resilience without Rewriting Bits in Packet Header}
Chiesa et al.~\cite{DBLP:journals/ton/ChiesaNMGMSS17} reveals that the $\left(k-1 \right) $\nobreakdash-resilience without rewriting bits in packet headers can be achieved in a $k$\nobreakdash-connected graph of $k\leq 5$ for static failures, but the ideal-resilience problem for a general $k$ still remains open. By Theorems~\ref{thm: 3-resilience}--\ref{thm: 5-resilience} and Lemma~\ref{lem: extension_k}, we will show that these conclusions can be extended to dynamic failures, i.e., the ideal $\left(k-1 \right) $\nobreakdash-resilience without rewriting bits  holds for dynamic failures  when $k\leq 5$.

\begin{theorem}
	Given  a $k$-connected graph $G$, with $k\leq 3$,  any circular-arborescences routing  is   $(k-1)$\nobreakdash-resilient  against dynamic failures. \label{thm: 3-resilience}
\end{theorem}
\begin{proof}
	In the following, we only give a proof for the case of $k=3$, which directly implies the proof for $k=2$. When $k=1$, the ideal resilience assumes no failures on a $1$-connected graph.
	There must be a set of three arc-disjoint arborescences $\mathcal{T}=\left\lbrace T_1, T_2, T_3 \right\rbrace$ of $G$. Let $H_F$ denote a meta-graph when $F$ is static, and let $h\subseteq H_F$ be a tree-component contained in $H_F$.

	If $h$ is a single node, then the arborescence $T\in \mathcal{T}$, represented by  the node $V\left( h\right) $ has no any failure even if $F$ denotes a set of dynamic failures.
	
	If $h$ contains two nodes and one edge, denoted by $\left\lbrace \mu_i, \mu_j\right\rbrace $,  then there are two arborescences $T_i, T_j\in \left\lbrace T_1, T_2, T_3 \right\rbrace $ that share the failure $e\in F$. It further implies that $T_i$ and $T_j$ are both good arborescences, s.t., bouncing on the failure $e$ of $T_i$ to $T_j$ will reach destination $t$ without seeing any failure along $T_j$ and vice versa, since $T_i$ and $T_j$ share a unique failure $e\in F$. Moreover, due to $k=3$,  it implies either $j= (i+1)\mod 3$ or $i= (j+1) \mod 3$, for $i,j\in \{1,2,3\}$ and $i\neq j$.
	The ordering of an arbitrary  circular-arborescences routing on $\mathcal{T}$ can be either $\left\langle  T_1, T_2, T_3 \right\rangle $ or $\left\langle  T_1, T_3, T_2 \right\rangle $. Thus,  a circular-arborescences routing must contain a switching of either $T_i\to T_j$ or $T_j\to T_i$ after hitting the failure $e$, which equals to a well-bouncing from $T_i$ (resp., $T_j$) to $T_j$ (resp., $T_i$), leading to arriving at $t$ uninterruptedly, even if $F$ is dynamic.
	
	If $h$ has $ \left|V \left( h\right)\right| =3$ and $ \left|E \left( h\right)\right| =2$, let  $V \left( h\right)=\left\lbrace \mu_i, \mu_j,\mu_\ell \right\rbrace $ and   $E\left( h\right)=\left\lbrace \{\mu_i,\mu_j\}, \{\mu_i,\mu_\ell\}\right\rbrace $. Since nodes $\left\lbrace \mu_i, \mu_j,\mu_\ell \right\rbrace $ denote  arborescences $T_i, T_j, T_\ell\in \left\lbrace T_1, T_2, T_3 \right\rbrace $, it implies that there exists a good arborescence  $T_i$ and  $(i+1)\mod 3 \in \{\ell, j\}$ for $i= 1,2,3$.  W.l.o.g., we assume $j= (i+1) \mod 3$.
	Now, after seeing the first failures $e_1$ on $T_i$, if $e_1\in E\left(T_j \right) $, then switching to the next arborescence indicates a  well-bouncing from $T_i$ to $T_j$, otherwise routing along the next arborescence $T_j$ will either arrive at $t$ or hit the second failure $e_2\in E\left( T_j\right) $. After seeing $e_2$ on $T_j$, the next arborescence will traverse along $T_\ell$ to hit $e_1\in E\left(T_\ell \right) $ again, which further leads to switching to the next arborescence $T_i$ again. Now, a canonical routing along $T_i$ can only hit the failure $e_2\in E\left( T_j\right) $, otherwise $T_i$ has a cycle containing $e_1$. After hitting $e_2$ on $T_i$, the next arborescence selected by a circular-arborescence routing will lead to a  well-bouncing from $T_i$ to $T_j$.
	
	For dynamic failures $F$, a canonical routing along an arborescence $T\in \mathcal{T}$ will arrive at $t$ directly if it does not hit any down link (failure) on $T$, but once hitting a  failure, the above analysis can be  applied, s.t., it eventually finds a well-bouncing to arrive at $t$.
\end{proof}
After studying the easier cases of $k\leq 3$, we present the complicated case of $k=4$ in Theorem~\ref{thm: 4-resilience}.

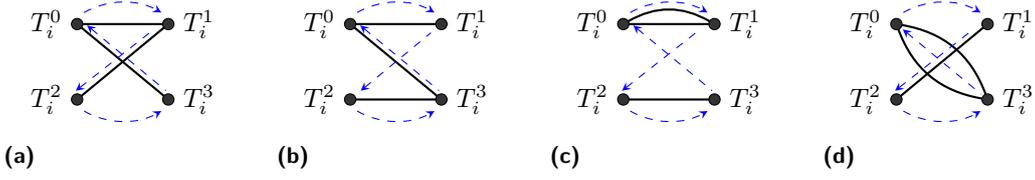
\begin{figure}[!t]
	
	\begin{subfigure}[t]{0.23\columnwidth}
		\centering
		\begin{tikzpicture}


\tikzset{mynode/.style={circle,draw,fill=black!80, outer sep=0pt,inner sep=1.5pt}};

\tikzset{source/.style={circle,draw,fill=green!80, outer sep=0pt,inner sep=2pt}};

\tikzset{destination/.style={circle,draw,fill=blue!50, outer sep=0pt,inner sep=2pt}};


\draw[] node[mynode,  label=left:$T^0_{i}$]  (a) at (0.8,0.5)  {};

\draw[] node[mynode,  label=left:$T^2_i$]  (c) at (0.8,-0.5)  {};

\draw[] node[mynode,label=right:$T^1_i$]  (b) at (2,0.5)  {};
\draw[] node[mynode, label=right:$T^3_i$]  (d) at (2,-0.5)  {};


\draw[  thick,] (a) to  (b);
\draw[  thick,] (b) to  (c);
\draw[  thick,] (a) to  (d);

\draw [>=stealth, ->, color=blue, shorten >=3,shorten <=3,  dashed] ($(a)+(0,0.1)$)  to[bend left] node {} ($(b)+(0,0.1)$) ;
\draw[>=stealth, <-, color=blue, shorten >=3,shorten <=3,    dashed] ($(c)+(-0.1,0.05)$)  to node [sloped,above,midway] {} ($(b)+(-0.1,0)$);

\draw [>=stealth, ->, color=blue, shorten >=3,shorten <=3, dashed] ($(c)+(0,-0.1)$)  to[bend right] node {} ($(d)+(0,-0.1)$) ;
\draw[>=stealth, <-, color=blue, shorten >=3,shorten <=3,    dashed] ($(a)+(0.05,0.05)$)  to node [sloped,above,midway] {} ($(d)+(0.05,0.05)$);

\end{tikzpicture}
		\caption{}
		\label{subfig:3-cases-1}
	\end{subfigure}
	\hfill
	\begin{subfigure}[t]{0.23\columnwidth}
		\centering
		\begin{tikzpicture}

	
	\tikzset{mynode/.style={circle,draw,fill=black!80, outer sep=0pt,inner sep=1.5pt}};
	
	\tikzset{source/.style={circle,draw,fill=green!80, outer sep=0pt,inner sep=2pt}};
	
	\tikzset{destination/.style={circle,draw,fill=blue!50, outer sep=0pt,inner sep=2pt}};
	

	\draw[] node[mynode,  label=left:$T^0_{i}$]  (a) at (0.8,0.5)  {};

	\draw[] node[mynode,  label=left:$T^2_i$]  (c) at (0.8,-0.5)  {};

	\draw[] node[mynode,label=right:$T^1_i$]  (b) at (2,0.5)  {};
	\draw[] node[mynode, label=right:$T^3_i$]  (d) at (2,-0.5)  {};
	

	\draw[  thick,] (a) to  (b);
	\draw[  thick,] (c) to  (d);
	\draw[  thick,] (a) to  (d);

	\draw [>=stealth, ->, color=blue, shorten >=3,shorten <=3,  dashed] ($(a)+(0,0.1)$)  to[bend left] node {} ($(b)+(0,0.1)$) ;
	\draw[>=stealth, <-, color=blue, shorten >=3,shorten <=3,   dashed] (c)  to node [sloped,above,midway] {} (b);
	
	\draw [>=stealth, ->, color=blue, shorten >=3,shorten <=3,  dashed] ($(c)+(0,-0.1)$)  to[bend right] node {} ($(d)+(0,-0.1)$) ;
	\draw[>=stealth, <-, color=blue, shorten >=3,shorten <=3,    dashed] ($(a)+(0.05,0.05)$)  to node [sloped,above,midway] {} ($(d)+(0.05,0.05)$);
	
%
%
%
%





\end{tikzpicture}
		\caption{}
		\label{subfig:3-cases-2}
		
	\end{subfigure}
	\hfill
	\begin{subfigure}[t]{0.23\columnwidth}
		\centering
		\begin{tikzpicture}

	
	\tikzset{mynode/.style={circle,draw,fill=black!80, outer sep=0pt,inner sep=1.5pt}};
	
	\tikzset{source/.style={circle,draw,fill=green!80, outer sep=0pt,inner sep=2pt}};
	
	\tikzset{destination/.style={circle,draw,fill=blue!50, outer sep=0pt,inner sep=2pt}};
	

	\draw[] node[mynode,  label=left:$T^0_{i}$]  (a) at (0.8,0.5)  {};

	\draw[] node[mynode,  label=left:$T^2_i$]  (c) at (0.8,-0.5)  {};

	\draw[] node[mynode,label=right:$T^1_i$]  (b) at (2,0.5)  {};
	\draw[] node[mynode, label=right:$T^3_i$]  (d) at (2,-0.5)  {};
	

	\draw[  thick,] (a) to  (b);
	\draw[  thick,] (a) to[bend left] node {}  (b);
	\draw[  thick,] (c) to  (d);

	\draw [>=stealth, ->, color=blue, shorten >=3,shorten <=3,  dashed] ($(a)+(0,0.1)$)  to[bend left] node {} ($(b)+(0,0.1)$) ;
	\draw[>=stealth, <-, color=blue, shorten >=3,shorten <=3,    dashed] ($(c)+(-0.1,0.05)$)  to node [sloped,above,midway] {} ($(b)+(-0.1,0)$);
	
	\draw [>=stealth, ->, color=blue, shorten >=3,shorten <=3, dashed] ($(c)+(0,-0.1)$)  to[bend right] node {} ($(d)+(0,-0.1)$) ;
	\draw[>=stealth, <-, color=blue, shorten >=3,shorten <=3,    dashed] ($(a)+(0.05,0.05)$)  to node [sloped,above,midway] {} ($(d)+(0.05,0.05)$);

\end{tikzpicture}
		\caption{}
		\label{subfig:3-cases-3}
	\end{subfigure}
	\hfill
	\begin{subfigure}[t]{0.23\columnwidth}
		\centering
		\begin{tikzpicture}

	
	\tikzset{mynode/.style={circle,draw,fill=black!80, outer sep=0pt,inner sep=1.5pt}};
	
	\tikzset{source/.style={circle,draw,fill=green!80, outer sep=0pt,inner sep=2pt}};
	
	\tikzset{destination/.style={circle,draw,fill=blue!50, outer sep=0pt,inner sep=2pt}};
	

	\draw[] node[mynode,  label=left:$T^0_{i}$]  (a) at (0.8,0.5)  {};

	\draw[] node[mynode,  label=left:$T^2_i$]  (c) at (0.8,-0.5)  {};

	\draw[] node[mynode,label=right:$T^1_i$]  (b) at (2,0.5)  {};
	\draw[] node[mynode, label=right:$T^3_i$]  (d) at (2,-0.5)  {};
	

	\draw[  thick,] (a) to[bend left=30]  (d);
	\draw[  thick,] (b) to  (c);
	\draw[  thick,] (a) to[bend right=30]  (d);

	\draw [>=stealth, ->, color=blue, shorten >=3,shorten <=3,  dashed] ($(a)+(0,0.1)$)  to[bend left] node {} ($(b)+(0,0.1)$) ;
	\draw[>=stealth, <-, color=blue, shorten >=3,shorten <=3,    dashed] ($(c)+(-0.1,0.05)$)  to node [sloped,above,midway] {} ($(b)+(-0.1,0)$);
	
	\draw [>=stealth, ->, color=blue, shorten >=3,shorten <=3, dashed] ($(c)+(0,-0.1)$)  to[bend right] node {} ($(d)+(0,-0.1)$) ;
	\draw[>=stealth, <-, color=blue, shorten >=1,shorten <=1,    dashed] (a)  to node [sloped,above,midway] {} (d);

\end{tikzpicture}
		\caption{}
		\label{subfig:3-cases-4}
	\end{subfigure}
	\caption{An illustration of main ideas to prove Theorem~\ref{thm: 4-resilience}.  When each arborescence $T\in \mathcal{T}$ with $\mathcal{T}=\left\lbrace T_1, T_2, T_3, T_4 \right\rbrace$ contains at least one failure in $F$ of $\left| F\right| \leq 3$, then its meta-graph $H_F=\left(V_F,  E_F\right) $ can be represented by one of these four subfigures in Fig.~\ref{fig:4-resilience-proof-idea}, where each node $T^j_i\in V_F$, $0\leq j\leq 3 $ and $1\leq i\leq 4$, denotes an arborescence $T_{\left( i+j\right) \mod 4}\in \mathcal{T}$, and each edge $\left\lbrace T^j_i, T^\ell_i\right\rbrace \in E_F$ (solid line) represents a failure in $F$, which is shared by two arborescences $T^j_i\in \mathcal{T}$ and $T^\ell_i\in \mathcal{T}$. The circular-arborescence routing  with the ordering $ \left\langle T_1, T_2, T_3, T_4  \right\rangle $, denoted by dashed (blue) arcs, always includes a bouncing from $T^j_i\in V_F$ to $T^\ell_i\in V_F$, where $T^\ell_i$ has a degree of one in $H_F$, indicating a potentiality of well-bouncing. After a circular-arborescence routing switching from $T^j_i\in V_F$ to $T^\ell_i\in V_F$, a canonical routing along $T^\ell_i$ might not arrive at the destination $t$ directly, even if the arc $\left(T^j_i, T^\ell_i \right) $ implies a well-bouncing, since the current failure confronted during a canonical routing along $T^j_i\in V_F$ may be different from the right failure that leads to the well-bouncing. However, we can prove that the routing eventually hits the right failure of the well-bouncing to approach $t$ by repeating the circular-arborescence routing of $ \left\langle T_1, T_2, T_3, T_4  \right\rangle $ at most two times.} \label{fig:4-resilience-proof-idea}
\end{figure}
\begin{theorem}\label{thm: 4-resilience}
	For a $k$-connected graph, with $k=4$, we can compute four arc-disjoint arborescences $\left\lbrace T_1, T_2, T_3, T_4 \right\rbrace$, s.t., $T_1$ and $T_3$ (resp., $T_2$ and $T_4$) are edge-disjoint by Lemma~\ref{lem: arc-disjoint-arborescences}. Then, the circular-arborescence routing (Definition~\ref{def: circular-arborescence} in Appendix~\ref{sec: techniques})  with the ordering $ \left\langle T_1, T_2, T_3, T_4  \right\rangle $ is $3$-resilient for dynamic failures.
\end{theorem}
\begin{proof}
	For  ease of understanding,  we first use  Fig.~\ref{fig:4-resilience-proof-idea}  to illustrate the main ideas of this proof and expand proof details in the following.
	
	Let $h$ be a tree in a meta-graph $H_F$ when $F$ is static. We note that $H_F$ is a bipartite graph $H_F=\left( V^1_F\cup V^2_F, E_F \right) $, where $V^1_F=\{\mu_1,\mu_3\}$ and  $V^2_F=\{\mu_2,\mu_4\}$.
	
	If $\left| V\left(h \right) \right| =1$, a canonical routing on $T\in \mathcal{T}$, which is denoted by the node of $h$, will not hit any failure before reaching $t$.
	
	If $\left| V\left(h \right) \right| \leq 3$, there must be a good arborescence $T_i$, denoted by a node in $h$, s.t., selecting the next $T_{\left( i+1\right) \mod 4}$ of $T_i$  by following the order  $ \left\langle T_1, T_2, T_3, T_4  \right\rangle $ can result in a well-bouncing from $T_i$. The details of the same arguments can be found in  the proof of Theorem~\ref{thm: 3-resilience}.
	
	If  $\left| V\left(h \right) \right| = 4$ and  $\left| E\left(h \right) \right| = 3$, then  $h$ must be a graph generated by removing an edge $\{\mu_{\ell}, \mu_f\}$ from the complete bipartite graph $H'_F=\left( V^1_F\cup V^2_F, E'_F \right)$, where $E'_F=\left\lbrace \{\mu_i, \mu_j\}:  i\in \{1,3\} \text{ and } j\in \{2,4\}\right\rbrace $. W.l.o.g., we can further assume $f=\left(  \ell+1\right)  \mod 4$ for $\ell \in \{1, 2, 3,4\}$, which implies that the node $\mu_{\ell}$ (resp., $\mu_f$) has the degree of one in $h$ and the arborescence $T_\ell$ (resp., $T_f$) only contains one dynamic failure. Let $ i' $ satisfy $\ell = \left(i'+1 \right) \mod 4 $ for $i' \in \{1,2, 3,4\}$. Now, it is clear that the arborescence $T_{i'}$ is a good arborescence and the bouncing from $T_{i'}$ to $T_\ell$ on the failure $e$ shared by $T_{i'}$ and $T_\ell$ is a well-bouncing. By our definition $\ell = \left(i'+1 \right) \mod 4 $ ,    $T_\ell$ is also the next arborescence of $T_{i'}$ for the circular-arborescences routing with the order $ \left\langle T_1, T_2, T_3, T_4  \right\rangle $. In other words, the circular-arborescences routing of $ \left\langle T_1, T_2, T_3, T_4  \right\rangle $ must include the bouncing from $T_{i'}$ to $T_\ell$. Starting a canonical routing along arborescence $T_{i'}$, we first hit the failure $e_1$. If $e_1$ is also shared by $T_\ell$, then bouncing on $e_1$ from $T_{i'}$ to $T_\ell$  is already a well-bouncing. Otherwise, we switch to $T_\ell$ to do a canonical routing along $T_\ell$ and hit the failure $e_2\in F$ on $T_\ell$. Clearly, here $e_2$ must be different from $e_1\in F$ since only one failure on $T_\ell$. After hitting $e_2$ on $T_\ell$, we switch to $T_f$, where $f= \left(\ell +1 \right) \mod 4$, and hit a failure $e_3\in F$ on $T_f$. Clearly, $e_3$ is not $e_1$ since $T_f$ and $T_{i'}$ are edge-disjoint. Due to  $\left\lbrace \mu_{\ell}, \mu_f\right\rbrace \notin E\left(h \right) $, it implies $T_f$ and $T_\ell$ cannot share any failure  in $F$, indicating  that $e_3\in E\left( T_f\right) $, $e_2\in E\left( T_\ell\right) $ and $e_3\neq e_2$. After seeing $e_3$ on $T_f$, we switch to the arborescence $T_j$, where $i= (f+1)\mod 4$ and $i'= (j+1)\mod 4$. If a canonical routing along $T_j$ cannot reach $t$, then a failure $e\in F$ on $T_j$ is confronted. Since $T_j$ and $T_\ell$ is edge-disjoint and $e_2\in E\left( T_\ell\right) $, then $e\neq e_2$. It further implies that $e=e_1$, otherwise $e=e_3$ leads to a loop containing $e_3$ on $T_{j}$. After seeing $e_1$ on $T_j$, the next arborescence switches to $T_{i'}$ and the  failure that can be hit on $T_{i'}$ must be $e_2$  since we start the canonical routing along $T_{i'}$  on the failure $e_1$ shared by $T_j$ and $T_{i'}$ is a directed tree. As the failure $e_2$ shared by $T_{i'}$ and $T_\ell$, the bouncing from  $T_{i'}$ to  $T_\ell$ on $e_2$ will be a well-bouncing to reach $t$ by a canonical routing along $T_\ell$.
\end{proof}
Chiesa et al.~\cite[Lemma~$7$]{DBLP:journals/ton/ChiesaNMGMSS17} show that  Lemma~\ref{lem: extension_k} holds for static failures. Next, we show that  Lemma~\ref{lem: extension_k}  is also true for dynamic failures.

\begin{lemma}
	\label{lem: extension_k}
	Given $k$ arc-disjoint arborescences $\mathcal{T}=\left\lbrace T_1, \ldots, T_k \right\rbrace$ of a graph $G$, if  a circular-arborescence routing on the first $k-1$ arborescences  $\mathcal{T}_{k-1}=\left\lbrace T_1, \ldots, T_{k-1} \right\rbrace$ is $\left(c-1 \right) $-resilient against dynamic failures with $c<k$, then there exists a $c$-resilient routing scheme in $G$.
\end{lemma}
\begin{proof}
	Chiesa et al.~\cite[Lemma~$7$]{DBLP:journals/ton/ChiesaNMGMSS17}  introduce a routing scheme $R$ as follows: a packet originated at a node $v\in V$ is first routed along the last arborescence $T_k$, and if a failure $\left(x,y \right) $ is hit along $T_k$ at $x$, then  it switches to a circular-arborescence routing based on arborescences $\left\lbrace T_1, \ldots, T_{k-1} \right\rbrace$ starting from the node $x$ along the arborescence $T'\in \left\lbrace T_1, \ldots, T_{k-1} \right\rbrace$ that contains $\left(y,x \right) $.
	
	Next, we show that $R$ is $c$-resilient against dynamic failurs if the circular-arborescence routing based on arborescences $\left\lbrace T_1, \ldots, T_{k-1} \right\rbrace$ (the starting arborescence is arbitrary) is $\left( c-1\right) $-resilient with $c<k$. If we meet the failure $\left( x,y\right) $ on $T_k$, then there is at most one arborescence $T'\in \mathcal{T}_{k-1}=\left\lbrace T_1, \ldots, T_{k-1} \right\rbrace$ that contains $\left(y,x\right) $. By   Theorem~\ref{thm: good-arborescence}, there must be at least one good arborescence in $\mathcal{T}_{k-1}$ for any $c-2$ dynamic failures $F'\subseteq E\left( \mathcal{T}_{k-1}\right) \setminus \{x,y\}$. If $T'$ is not the good arborescence, then the original good  arborescence  in $\mathcal{T}_{k-1}$ under $F'$  remain the same for failures $F=F'\cup \{x,y\}$, implying that circular-arborescence routing based on $\mathcal{T}_{k-1}$  can still converge to $t$ under $F$, where $\left| F\right| =c$ and  $\left| F'\right| =c-1$.

	 Now, we suppose that  $T'$ is the  good arborescence in $\mathcal{T}_{k-1}$ under $F'$. If $\left(y,x \right)\in E\left(T' \right)  $ is the highest failure, then we reach $t$ directly when bouncing from $T_k$ to $T'$. Next, the remaining scenario is to show that any circular-arborescence routing starting from $T'$ for a packet originated at the node $x$ will always lead the packet to $t$. Since the circular-arborescence  routing on $\mathcal{T}_{k-1}$ is $\left( c-1\right) $-resilient, then we only need to show that the packet originated at $x$ starting from $T'$ can reach $t$ before hitting $(y,x)$ under any $c-1$ dynamic failures $F'\subseteq \left( \mathcal{T}_{t-1} \right) \setminus \{x,y\}$.  By assuming $\{x,y\}$ is not failed,  then the circular-arborescence  routing on $\mathcal{T}_{k-1}$, under $F'$, starting from  $T'$ at the node $x$ cannot visit the arc $\left(y,x \right) $ along $T'$ before reaching $t$, otherwise a routing loop starting at $x$ and going back to $x$ again exists. Therefore, $R$ is $c $-resilient.
\end{proof}
After establishing Lemma~\ref{lem: extension_k}, by Theorem~\ref{thm: 4-resilience}, we can extend the ideal $(k-1)$\nobreakdash-resilience from  $k\leq 4$ to $k=5$.

\begin{theorem} 
	For any $5$-connected graph, there exists a $4$-resilient routing scheme against dynamic failures. \label{thm: 5-resilience}
\end{theorem}
\begin{proof}
	For a  $5$-connected graph, we can compute five arc-disjoint arborescences  $\left\lbrace T_1, \ldots, T_{5} \right\rbrace$. By Theorem~\ref{thm: 4-resilience}, we can find a circular-arborescence routing on $\left\lbrace T_1, \ldots, T_{4} \right\rbrace$, which is $3$-resilient against dynamic failures. Then, by Lemma~\ref{lem: extension_k}, we can further find a routing scheme based on $\left\lbrace T_1, \ldots, T_{5} \right\rbrace$, which is $4$-resilient against dynamic failures.
\end{proof}
By Theorem~\ref{thm: k/2-resilience}, we show that ideal resilience can be attained in an arbitrary $k$-connected graph if the number of failures is at most half of the edge connectivity.
\begin{theorem}
	For any $k$-connected graph, there exists
	a $\left\lfloor\frac{k}{2} \right\rfloor$\nobreakdash-resilient routing scheme against dynamic failures. \label{thm: k/2-resilience}
\end{theorem}
\begin{proof}
	For a  $k$-connected graph  $G$, there are $k$ arc-disjoint arborescences  $\mathcal{T}=\left\lbrace T_1, \ldots, T_{k} \right\rbrace$.  For a set $F$ of $\left\lfloor \frac{k}{2} \right\rfloor-1$ dynamic failures on $G$, there must be at one arborescence $T\in \mathcal{T}$, s.t., $T$ does not contain any failure in $F$. It implies that every circular-arborescence routing is $\left( \left\lfloor \frac{k}{2} \right\rfloor-1\right) $-resilient. Then, by Lemma~\ref{lem: extension_k}, there must be a $\left\lfloor \frac{k}{2} \right\rfloor$-resilient routing on $G$ against dynamic failures.
\end{proof}

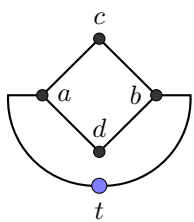
\begin{figure}[t]
	\centering
	\begin{tikzpicture}

	
	\tikzset{mynode/.style={circle,draw,fill=black!80, outer sep=0pt,inner sep=1.5pt}};
	
	\tikzset{source/.style={circle,draw,fill=green!80, outer sep=0pt,inner sep=2pt}};
	
	\tikzset{destination/.style={circle,draw,fill=blue!50, outer sep=0pt,inner sep=2pt}};
	
	
	\draw[] node[mynode,  label=right:$a$]  (a) at (-0.75,0)  {};
	
	\draw[] node[mynode,  label=left:$b$]  (b) at (0.75,0)  {};
	
	\draw[] node[mynode,  label=above:$c$]  (c) at (0,0.75)  {};
	
	\draw[] node[mynode,  label=above:$d$]  (d) at (0,-0.75)  {};
	
	\draw[] node[destination,  label=below:$t$]  (t) at (0,-1.2)  {};
	
		\draw[thick] (a) to[] (c);
		\draw[thick] (c) to[] (b);
		\draw[thick] (b) to[] (d);
		\draw[thick] (d) to[] (a);
		
			\draw[thick] (a) -- (-1.2,0) to[bend right=45] (t);
			\draw[thick] (b) -- (1.2,0) to[bend left=45] (t);

\end{tikzpicture}
	\vspace{-1mm}
	\caption{Counter-example for achieving $1$-resilience against dynamic failures in a $2$-connected graph $G$ when each node must employ a link-circular routing function.  When each node uses a link-circular routing, it  has only two possible orderings for its neighbors, i.e., clockwise and counter-clockwise for the shown drawing. For example, the clockwise and counter-clockwise orderings for $a$ are $\left<t,c,d\right>$ and  $\left<t,d,c\right>$, respectively. If $a$ and $b$ use the clockwise (resp., counter-clockwise) orderings, given a dynamic failure $F=\{c,b\}$ (resp., $F=\{b,d\}$) and a packet originated at $c$ (resp., $d$),  a forwarding loop: $\left(c,a,d,b,c\right) $ (resp., $\left(d,a,c,b,d \right) $) occurs, where $\{c,b\}$ (resp., $\{b,d\}$) is down only when the packet is originated for initial sending  but recovered afterwards. However, if $b$ and $a$ use  forwarding functions of clockwise and counter-clockwise orderings, respectively, and the node $c$ send its original packet to $v\in \{a,b\}$, the static failure $F=\left\lbrace v,t\right\rbrace $, implies a  forwarding loop: $\left(c,v,d, v' \right) $ with $v'=\{a,b\}\setminus v$. Analogous arguments can be given if  $a$ and $b$ reverse the orderings of their routing functions respectively.}
	\label{fig:fig1}
	\vspace{-2mm}
\end{figure}
Chiesa et al.~\cite[Theorem~$15$]{DBLP:journals/ton/ChiesaNMGMSS17} show that $2$-resilience cannot be achieved in a $3$-connected graph for static failures if each node employs a link-circular routing function. In Theorem~\ref{thm: no-1-resilience}, we can prove an even stronger conclusion for dynamic failures, where even the $1$\nobreakdash-resilience on a $2$\nobreakdash-connected graph cannot stand if  link-circular routing is applied on each node. We will prove  Theorem~\ref{thm: no-1-resilience} by giving a counter-example, as illustrated in  Fig.~\ref{fig:fig1}.

\begin{algorithm}[t]
	\caption{\textsc{HDR-Log}-K-\textsc{Bits}~{\cite[Algorithm~1]{DBLP:journals/ton/ChiesaNMGMSS17}}}\label{alg:log_k bits}
	\KwData{A set of $k$ arc-disjoint arborescences  $\mathcal{T}=\left\lbrace T_1, \ldots, T_k \right\rbrace $ and a destination $t$}
	Let $T_i\in \mathcal{T}$ be the first arborescence that is used to route a packet\;
	Set $current\_id:=i$\;
	\While{the packet is not delivered to $t$}{
		canonical routing along $T_i$ until reaching $t$ or hitting a failure $e\in F$\;
		\If{ $e\in F$ is shared by an arborescence $T_j$, $i\neq j$}{
			\eIf{$current\_id\neq i$}{$current\_id= \left( current\_id+ 1\right) \mod k$\; $i:=current\_id$\;}{$i:=j$}
		}
	}

\end{algorithm}
\begin{algorithm}[t]
	\KwData{A set of $k$ arc-disjoint arborescences  $\mathcal{T}=\left\lbrace T_1, \ldots, T_k \right\rbrace $ and a destination $t$}
	Set $i:=1$\;
	\While{the packet is not delivered to $t$}{
		canonical routing along $T_i$ until reaching $t$ or hitting a failure $e\in F$\;
		\If{ $e\in F$ is shared by an arborescence $T_j$, $i\neq j$}{
			Bounce and route along DFS traversal in $T_j$\;
			\If{the routing hits a failure $e'\in F$ on $T_j$}{Route back to the failure $e$ to determine $T_i$ by reversing DFS traversal   on $T_j$\;}
		}
		Set $i:= \left( i+1\right) \mod k$\;
	}

	\caption{HDR-$3$-\textsc{Bits}~{\cite[Algorithm~2]{DBLP:journals/ton/ChiesaNMGMSS17}}}\label{alg: three_bits}
\end{algorithm}
\begin{theorem}\label{thm: no-1-resilience}
	There exists a $2$-connected graph $G$ for which no $1$-resilient routing scheme  against dynamic failure can exist if each node  $v\in V\left( G\right) $  uses  a link-circular routing function. 
\end{theorem}
\begin{proof}
	We construct a graph $G=\left(V,E \right) $ as shown in Fig.~\ref{fig:fig1}, where $V=\left\lbrace a,b,c,d,t\right\rbrace $ and \[E=\left\lbrace\left\lbrace a,c \right\rbrace, \left\lbrace a,d \right\rbrace, \left\lbrace a,t \right\rbrace, \left\lbrace b,c \right\rbrace,  \left\lbrace b,d \right\rbrace, \left\lbrace b,d \right\rbrace    \right\rbrace \,.\] Let $F\in E$ denote a dynamic failure (link) in $G$. For an arbitrary link-circular routing  with a destination $t\in V$, for each node $v\in V\setminus\left\lbrace t\right\rbrace $, if  $\Delta_{G\setminus F}\left( v\right)=2$, then  $\pi_v\left(x \right)=y $ and $\pi_v\left(y \right)=x$, where $N_{G\setminus F}\left( v\right)  =\left\lbrace x, y\right\rbrace \subset V$ denotes neighbors of $v$ in $G\setminus F$, and  if  $\Delta_{G\setminus F}\left( v\right)=1$, then $\pi_v\left(x \right)=x $, where $N_{G\setminus F}\left( v\right)  = x\in V.$   For the current drawing of $G$ as shown in Fig.~\ref{fig:fig1}, if  $\Delta_{G\setminus F}\left( v\right)=3$,  a link-circular routing on  $a$ (resp., $b$) must have either a clockwise ordering of $N_{G\setminus F}\left( a\right)$ (resp., $N_{G\setminus F}\left( b\right)$), i.e., $\left<t,c,d\right>$ (resp., $\left<c,t,d\right>$),  or the anticlockwise ordering of $N_{G\setminus F}\left( a\right)$ (resp., $N_{G\setminus F}\left( b\right)$), i.e., $\left<c,t,d\right>$ (resp., $\left<c,d,t\right>$). For each node $v\in \{b,d\}$, if $\Delta_{G\setminus F}\left( v\right) =1$, then the packet originated at $v$ must be sent through the unique link incident on $v$ in $G\setminus F$.
	
	Let us first consider the following two scenarios, where the routing functions on $a$ and $b$ have the same type of ordering of their neighboring nodes.
	\begin{itemize}
		\item Routing functions on $a$ and $b$ are both clockwise, $F=\{c,b\}$, a packet originated at $c$;
		\item Routing functions on $a$ and $b$ are both anti-clockwise, $F=\{b,d\}$, a packet originated at $d$.
	\end{itemize}
	If  $F$ is  down  only when $c$ (resp., $d$) sends a packet originated at itself, but  up in other time slots, we can easily verify that each of above two scenarios leads to a routing loop along the cycle $\{a,b,c,d\}$.
	
	Next, we deal with the cases, where $a$ and $b$ have the different type of ordering of their neighboring nodes. First, let $a$ take anti-clockwise but let $b$ take clockwise respectively.  If $c$ sends its original packet to $v\in \{a,b\}$ when $\Delta_{G\setminus F}\left( c\right)=2$, then $F=\{v,t\}$ acting as a static failure  can lead to a routing loop on $\{a,b,c,d\}$ directly. By symmetry, a similar result can be proved in much the same way when we reverse the type of ordering on $a$ and $b$ respectively.
	
	It is easy to see that our discussion has covered all possible link-circular routing functions on $\{a,b,c,d\}$. Thus, we can conclude no $1$-resilient routing in $G$, employing a link-circular routing function on each node,  against dynamic failure. 
\end{proof}
\subsection{Ideal Resilience by Packet Header Rewriting}
Given the  results of the ideal $\left(k-1 \right) $\nobreakdash-resilience of $k\leq 5$, the question arises, whether   $\left(k-1 \right) $\nobreakdash-resilience is feasible for any $k$. The previous work by Chiesa et al.~\cite{DBLP:journals/ton/ChiesaNMGMSS17}  only showed that the $\left(k-1 \right) $\nobreakdash-resilience for a general $k$ is possible by rewriting $\left\lceil \log k \right\rceil$ or three  bits in packet headers under static failures. We will show that the \textsc{HDR-Log}-K-\textsc{Bits}~\cite[Algorithm~1]{DBLP:journals/ton/ChiesaNMGMSS17} algorithm also works for dynamic failures, but HDR-$3$-\textsc{Bits}~\cite[Algorithm~2]{DBLP:journals/ton/ChiesaNMGMSS17} algorithm becomes ineffective for dynamic failures. The pseudocodes of \textsc{HDR-Log}-K-\textsc{Bits}~\cite[Algorithm~1]{DBLP:journals/ton/ChiesaNMGMSS17}  and HDR-$3$-\textsc{Bits}~\cite[Algorithm~2]{DBLP:journals/ton/ChiesaNMGMSS17}  are presented by Algorithm~\ref{alg:log_k bits} and  Algorithm~\ref{alg: three_bits} respectively.
\begin{theorem}
	For a $k$-connected graph, Algorithm~\ref{alg:log_k bits} (\textsc{HDR-Log}-K-\textsc{Bits}~\cite[Algorithm~1]{DBLP:journals/ton/ChiesaNMGMSS17}) is a $\left(k-1 \right) $\nobreakdash-resilient routing against dynamic failures by rewriting at most $\left\lceil \log k \right\rceil$ bits in the packet headers. \label{thm: log-k-bits-reilsiency}
\end{theorem}
\begin{proof}
	By Theorem~\ref{thm: good-arborescence}, there must be a good arborescence against $k-1$ dynamic failures. In Algorithm~\ref{alg:log_k bits}, the while loop conducts a circular-arborescence routing on $\left\langle T_1, \ldots,T_k \right\rangle $ by maintaining $current\_id= i$. If the canonical routing on the current arborescence $T_i$ hits a failure $e\in F$, which is shared by an arborescence $T_j$, then bouncing from $T_i$ to $T_j$ occurs and $T_j$ becomes the next arborescence for canonical routing. If $T_i$ is a good arborescence, the packet reaches $t$ along $T_j$, otherwise it switches to the circular-arborescence routing again by setting the next arborescence as $T_{\left( \left( i+ 1\right) \mod k\right) }$. We need the variable $current\_id$ to keep the index of the current arborescence in $\left\langle T_1, \ldots,T_k \right\rangle $ when bouncing  on a failure occurs, and we need $\lceil \log k \rceil$ bits to store  $current\_id$.
\end{proof}
In Theorem~\ref{thm: 3bits-original}, we repeat the conclusion for the HDR-$3$-\textsc{Bits} algorithm under static failures by   Chiesa et al.~\cite{DBLP:journals/ton/ChiesaNMGMSS17,robroute16infocom}. We refer the reader to~\cite[Theorem~$5$]{robroute16infocom} for the proof details.
\begin{theorem}[{\cite[Theorem~$5$]{robroute16infocom}}] 
	For a $k$-connected graph,  Algorithm~\ref{alg: three_bits} (HDR-$3$-\textsc{Bits}~\cite[Algorithm~2]{DBLP:journals/ton/ChiesaNMGMSS17}) is a $\left(k-1 \right) $\nobreakdash-resilient routing against static failures by rewriting at most $3$ bits in  packet headers.  \label{thm: 3bits-original}
\end{theorem}

\begin{figure*}[t!]
	\centering
			\resizebox{\textwidth}{!}{

	\begin{tikzpicture}

		\tikzset{mynode/.style={circle,draw,fill=black!80, outer sep=0pt,inner sep=1.5pt}};

\tikzset{source/.style={circle,draw,fill=green!80, outer sep=0pt,inner sep=2pt}};

\tikzset{destination/.style={circle,draw,fill=blue!50, outer sep=0pt,inner sep=2pt}};
		
\begin{scope}[shift={(0,0)}]


\draw[] node[mynode,  label=above:$d$]  (a) at (-0.75,0)  {};

\draw[] node[mynode,  label=above:$c$]  (b) at (0,0)  {};

\draw[] node[mynode,  label=above:$b$]  (c) at (0.75,0)  {};

\draw[] node[mynode,  label=above:$a$]  (d) at (1.5,0)  {};

\draw[] node[destination,  label=above:$t$]  (t) at (0.385,1)  {};
\draw[] node[mynode,  label=below:$x$]  (x) at (0.385,-1)  {};

\node[align=center] at (0.385,-1.8) {Graph $G$};

\draw[] (a) to (t);
\draw[] (c) to (t);
\draw[] (b) to (t);
\draw[] (d) to (t);

\draw[] (a) to (x);
\draw[] (c) to (x);
\draw[] (b) to (x);
\draw[] (d) to (x);

\draw[] (a) to[bend right=30] (c);
\draw[] (a) to[bend right=40] (d);
\draw[] (b) to[bend right=30] (d);
%
\draw[] (a) -- (b)-- (c) --(d);
\end{scope}

\begin{scope}[shift={(3,0)}]

		\draw[] node[mynode,  label=above:$d$]  (a) at (-0.75,0)  {};
		
		\draw[] node[mynode,  label=above:$c$]  (b) at (0,0)  {};
		
		\draw[] node[mynode,  label=above:$b$]  (c) at (0.75,0)  {};
		
		\draw[] node[mynode,  label=above:$a$]  (d) at (1.5,0)  {};
		
		\draw[] node[destination,  label=above:$t$]  (t) at (0.385,1)  {};
		\draw[] node[mynode,  label=below:$x$]  (x) at (0.385,-1)  {};

		\node[align=center] at (0.385,-1.8) {$T_1$};
		
		\draw[opacity=0.3] (a) to (t);
		\draw[opacity=0.3] (c) to (t);
		\draw[opacity=0.3] (b) to (t);
		\draw[opacity=0.3] (d) to (t);

		\draw[opacity=0.3] (a) to (x);
		\draw[opacity=0.3] (c) to (x);
		\draw[opacity=0.3] (b) to (x);
		\draw[opacity=0.3] (d) to (x);
		
		\draw[opacity=0.3] (a) to[bend right=30] (c);
		\draw[opacity=0.3] (a) to[bend right=40] (d);
		\draw[opacity=0.3] (b) to[bend right=30] (d);
		%
		\draw[opacity=0.3] (a) -- (b)-- (c) --(d);

		\draw[transform canvas={yshift=-0.5ex,}, red,->,>=stealth,shorten >=1,shorten <=1, thick] (d) --(t);
		\draw[transform canvas={xshift=-2,yshift=2}, red,->,>=stealth,shorten >=1,shorten <=2, thick] (c) --(d);
		\draw[transform canvas={xshift=-2,yshift=2}, red,->,>=stealth,shorten >=1,shorten <=2, thick] (a) to[bend right=40] (d);
		
		\draw[transform canvas={xshift=-2,yshift=0}, red,->,>=stealth,shorten >=0,shorten <=0, thick] (b) --(x);
		\draw[transform canvas={xshift=2,yshift=0}, red,->,>=stealth,shorten >=1,shorten <=0, thick] (x) --(c);
	\end{scope}

	\begin{scope}[shift={(6,0)}]

	\draw[] node[mynode,  label=above:$d$]  (a) at (-0.75,0)  {};
	
	\draw[] node[mynode,  label=above:$c$]  (b) at (0,0)  {};
	
	\draw[] node[mynode,  label=above:$b$]  (c) at (0.75,0)  {};
	
	\draw[] node[mynode,  label=above:$a$]  (d) at (1.5,0)  {};
	
	\draw[] node[destination,  label=above:$t$]  (t) at (0.385,1)  {};
	\draw[] node[mynode,  label=below:$x$]  (x) at (0.385,-1)  {};

	\node[align=center] at (0.385,-1.8) {$T_2$};
	
	\draw[opacity=0.3] (a) to (t);
	\draw[opacity=0.3] (c) to (t);
	\draw[opacity=0.3] (b) to (t);
	\draw[opacity=0.3] (d) to (t);

	\draw[opacity=0.3] (a) to (x);
	\draw[opacity=0.3] (c) to (x);
	\draw[opacity=0.3] (b) to (x);
	\draw[opacity=0.3] (d) to (x);
	
	\draw[opacity=0.3] (a) to[bend right=30] (c);
	\draw[opacity=0.3] (a) to[bend right=40] (d);
	\draw[opacity=0.3] (b) to[bend right=30] (d);
	%
	\draw[semitransparent] (a) -- (b)-- (c) --(d);

	\draw[transform canvas={yshift=0.5ex,}, green,->,>=stealth,shorten >=1,shorten <=0, thick] (a) --(t);

	\draw[transform canvas={yshift=-0.5ex,}, green,->,>=stealth,shorten >=-1,shorten <=2, thick] (x) --(a);
	
	\draw[transform canvas={xshift=-0.5, yshift=2}, green,->,>=stealth,shorten >=-1,shorten <=-1, thick] (d) --(c);
		\draw[transform canvas={xshift=-0.5, yshift=2}, green,->,>=stealth,shorten >=0,shorten <=-1, thick] (c) --(b);
				\draw[transform canvas={xshift=-0.5, yshift=2}, green,->,>=stealth,shorten >=0,shorten <=-2, thick] (b) --(a);
\end{scope}

\begin{scope}[shift={(12,0)}]

	\draw[] node[mynode,  label=above:$d$]  (a) at (-0.75,0)  {};
	
	\draw[] node[mynode,  label=above:$c$]  (b) at (0,0)  {};
	
	\draw[] node[mynode,  label=above:$b$]  (c) at (0.75,0)  {};
	
	\draw[] node[mynode,  label=above:$a$]  (d) at (1.5,0)  {};
	
	\draw[] node[destination,  label=above:$t$]  (t) at (0.385,1)  {};
	\draw[] node[mynode,  label=below:$x$]  (x) at (0.385,-1)  {};

	\node[align=center] at (0.385,-1.8) {$T_4$};
	
	\draw[opacity=0.3] (a) to (t);
	\draw[opacity=0.3] (c) to (t);
	\draw[opacity=0.3] (b) to (t);
	\draw[opacity=0.3] (d) to (t);

	\draw[opacity=0.3] (a) to (x);
	\draw[opacity=0.3] (c) to (x);
	\draw[opacity=0.3] (b) to (x);
	\draw[opacity=0.3] (d) to (x);
	
	\draw[opacity=0.3] (a) to[bend right=30] (c);
	\draw[opacity=0.3] (a) to[bend right=40] (d);
	\draw[opacity=0.3] (b) to[bend right=30] (d);
	%
	\draw[opacity=0.3] (a) -- (b)-- (c) --(d);

	\draw[transform canvas={xshift=-2}, blue,->,>=stealth,shorten >=0,shorten <=0, thick] (c) --(t);
	\draw[transform canvas={yshift=2}, blue,->,>=stealth,shorten >=0,shorten <=0, thick] (b) --(c);
	\draw[transform canvas={xshift=0,yshift=-2}, blue,<-,>=stealth,shorten >=0,shorten <=-1.5, thick] (a) to[bend right=40] (d);
	\draw[transform canvas={xshift=0,yshift=1.5}, blue,->,>=stealth,shorten >=0,shorten <=1, thick] (a) to[bend right=30] (c);
	\draw[transform canvas={xshift=2,yshift=0}, blue,->,>=stealth,shorten >=-2,shorten <=0, thick] (x) --(b);
\end{scope}		

\begin{scope}[shift={(9,0)}]

	\draw[] node[mynode,  label=above:$d$]  (a) at (-0.75,0)  {};
	
	\draw[] node[mynode,  label=above:$c$]  (b) at (0,0)  {};
	
	\draw[] node[mynode,  label=above:$b$]  (c) at (0.75,0)  {};
	
	\draw[] node[mynode,  label=above:$a$]  (d) at (1.5,0)  {};
	
	\draw[] node[destination,  label=above:$t$]  (t) at (0.385,1)  {};
	\draw[] node[mynode,  label=below:$x$]  (x) at (0.385,-1)  {};

	\node[align=center] at (0.385,-1.8) {$T_3$};
	
	\draw[opacity=0.3] (a) to (t);
	\draw[opacity=0.3] (c) to (t);
	\draw[opacity=0.3] (b) to (t);
	\draw[opacity=0.3] (d) to (t);

	\draw[opacity=0.3] (a) to (x);
	\draw[opacity=0.3] (c) to (x);
	\draw[opacity=0.3] (b) to (x);
	\draw[opacity=0.3] (d) to (x);
	
	\draw[opacity=0.3] (a) to[bend right=30] (c);
	\draw[opacity=0.3] (a) to[bend right=40] (d);
	\draw[opacity=0.3] (b) to[bend right=30] (d);
	%
	\draw[semitransparent] (a) -- (b)-- (c) --(d);

\draw[transform canvas={xshift=2}, orange,->,>=stealth,shorten >=1,shorten <=-1, thick] (b) --(t);
\draw[transform canvas={yshift=2}, orange,->,>=stealth,shorten >=-1,shorten <=-1, thick] (a) --(b);

\draw[transform canvas={xshift=0,yshift=-1.5}, orange,<-,>=stealth,shorten >=0,shorten <=-1, thick] (a) to[bend right=30] (c);
\draw[transform canvas={xshift=0,yshift=1.5}, orange,<-,>=stealth,shorten >=1,shorten <=0, thick] (b) to[bend right=30] (d);
\draw[transform canvas={xshift=2,yshift=0}, orange,->,>=stealth,shorten >=0,shorten <=0, thick] (x) --(d);

\end{scope}

	\end{tikzpicture}

}
	\caption{Counter-example for applying HDR-$3$-\textsc{Bits} (\hspace{1sp}\cite[Algorithm~2]{DBLP:journals/ton/ChiesaNMGMSS17}, presented in Algorithm~\ref{alg: three_bits}) against dynamic failures. For the $4$-connected graph $G=\left(V,E \right) $ and four arc-disjoint arborescences $\{T_1, \ldots,T_4\}$ as shown in this figure,  HDR-$3$-\textsc{Bits} algorithm will result in a routing loop for the dynamic failures: $F=\left\lbrace \{a,b\}, \{b,c\}, \{c,d\}\right\rbrace $. A packet originated at the node $x\in V$ will be routed along $T_1$ until hitting the first failure $(b,a)$, and then it is bounced to $T_2$ to follow the directed path $\left(b,c,d \right) $ until hitting the second failure $\left( c,d\right) $. Now, the packet starts at $c$ to do a reversing DFS traversal of $T_2$ to hit the failure $\left(c,b \right) $. Algorithm~\ref{alg: three_bits} will interpret $\left(c,b \right) $ as the first failure to believe that  the current arborescence $T_i$ is $T_4$. Thus, the algorithm will select the next arborescence of $T_i$ as $T_1$ instead of $T_2$ and the packet will follow the directed path $\left(c, x,b \right) $ on $T_1$ to hit the failure $\left(b,a \right) $ again, s.t., the same steps are repeated to generate a routing loop.}
	\label{fig:fig2}

\end{figure*}
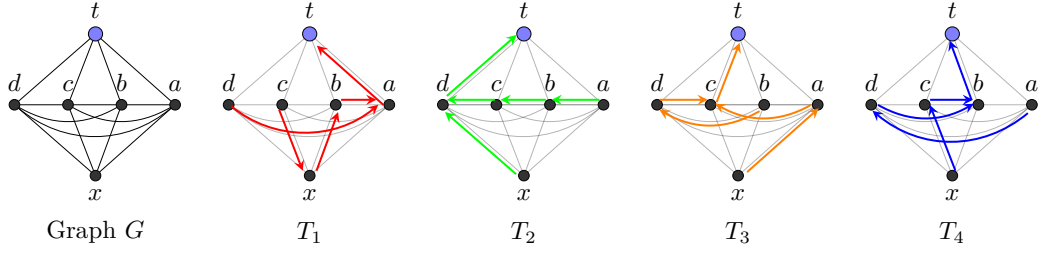
In the following, we introduce a counter-example for  the HDR-$3$\nobreakdash-\textsc{Bits} algorithm, which can lead to a forwarding loop even for three dynamic failures in a $4$-connected graph.
\begin{theorem}
	For a $k$-connected graph with $k\ge 4$,  Algorithm~\ref{alg: three_bits} (HDR-$3$-\textsc{Bits}~\cite[Algorithm~2]{DBLP:journals/ton/ChiesaNMGMSS17}) cannot be a $\left(k-1 \right) $\nobreakdash-resilient routing against dynamic failures by rewriting at most $3$ bits in  packet headers.  \label{thm: counter-example-3bits}
\end{theorem}
\begin{proof}
	We will show a counter-example for Algorithm~\ref{alg: three_bits}, which can result in a routing loop for three dynamic failures.

	As shown in Fig.~\ref{fig:fig2},  we can construct a  $4$-connected graph $G=\left(V,E \right) $ and four arc-disjoint arborescences $\left\lbrace T_1, T_2, T_3, T_4 \right\rbrace $ rooted at the node $t\in V$.
	
	Let  three  dynamic failures $F$ be $\{a,b\}$, $\{b,c\}$, and $\{c,d\}$. For a packet starting at the node $b$, it is first routed along $T_1$ to meet the first failure $\left( a,b\right) $. Since $\{a,b\}$ is shared by $T_1$ and $T_2$, we will bounce from $T_1$ to $T_2$ and start a DFS traversal along $T_2$ from  $b$, i.e., following the directed path $\left(b,c,d,t \right) $.
	
	When $\left(b,c \right) $ is up and  $\left(c,d \right) $ is failed, the DFS traverseal along $T_2$ will  stop at $c$ due to the second failure $\left( c,d\right) $. Now, if the reversing DFS traversal on $T_2$ from the node $c$ hits the failure  $\left(c,b \right) $, then Algorithm~\ref{alg: three_bits} will conclude that the first failure should be $\left(c,b \right) $ instead of $\left( a,b\right) $ and the current arborescence is $T_i=T_4$. According to   Algorithm~\ref{alg: three_bits}, we should shift the current arborescence $T_4$ to the next one, which is $T_1$ again. Starting at $c$ on $T_1$, the canonical routing along $T_1$ will go through $\left(c, x,b \right) $ to hit the failure $\left(b,a \right) $ again to repeat the previous routing loop.
\end{proof}
Although the HDR-$3$-\textsc{Bits} algorithm cannot be applied to dynamic failures anymore, we further illustrate that the algorithm can still work for semi-dynamic failures, where a link becomes permanently failed  once its state is down, in contrast to arbitrary link states of dynamic failures.

\begin{theorem}
	For a $k$-connected graph,  Algorithm~\ref{alg: three_bits} (HDR-$3$-\textsc{Bits}~\cite[Algorithm~2]{DBLP:journals/ton/ChiesaNMGMSS17}) is a $\left(k-1 \right) $\nobreakdash-resilient routing against semi-dynamic failures by rewriting at most $3$ bits in the packet headers. \label{thm: semi-dynamic-3bits}
\end{theorem}
\begin{proof}
	The pseudo-code of Algorithm~\ref{alg: three_bits} implies that Algorithm~\ref{alg: three_bits}  will not stop until packet reaches $t$. We can divide an execution of Algorithm~\ref{alg: three_bits} on any set of $\left( k-1\right) $ semi-dynamic failures into several phases. Start running  Algorithm~\ref{alg: three_bits} on an arbitrary arborescence, which is the initial phase, and when the packet observes two different states on a link $e$ during its traversing on arborescences, a new phase of the algorithm execution starts immediately. For a semi-dynamic failure,  it becomes a static failure after becoming  down for the first time. Thus, the $k-1$ semi-dynamic failures (links) indicates at most $k$ phases. Clearly, in the last phase, the packet cannot observe any link that can change its state, otherwise another new phase starts again and the current phase is not the last phase.
	
	Now, in the last phase, we show that Algorithm~\ref{alg: three_bits} will stop by sending the packet to the destination $t$.  Easy to note that, all failures that will be visited in the last phase must be already fixed as static failures in the beginning of the last phase, otherwise it is not the last phase yet. Thus, the last phase can be understood as the beginning time of running  Algorithm~\ref{alg: three_bits} for static failures.  Similar to  the analysis by Chiesa et al.~\cite{DBLP:journals/ton/ChiesaNMGMSS17}, by iterating on each arborescence in $\mathcal{T}$, the packet can finally find a good arborescence to reach $t$.
\end{proof}

\section{Perfect Resilience Against Dynamic Failures}\label{sec: perfect}
We devote this section to the $k$-resilience in a general graph. First, we show that the $1$\nobreakdash-resilience always exists against dynamic failures, but no $2$\nobreakdash-resilient source-matched routing anymore  for  dynamic failures, and finally demonstrate that the perfect $k$-resilience is impossible by rewriting $O(\log k)$ bits.
\begin{figure}[t]
	\centering
	\input{Figures/intuition_example.tex}
	\caption{Example of applying the  $2$\nobreakdash-resilient source-matched routing algorithm proposed by Dai et al.~\cite[Algorithm~$1$]{DBLP:conf/spaa/DaiF023} to a graph $G=\left(V,E \right) $ shown as bold lines without arrows in Fig.~\ref{fig:fig_SPAA} (\cite[Fig.~$1$]{DBLP:conf/spaa/DaiF023}) for the source-destination pair $\left(s,t \right) $ to obtain its kernel graph $\mathcal{G}$ by excluding  these four red bold lines: $\left\lbrace\{v_1,v_4\}, \{v_2,v_3\},\{u_1,u_4\}, \{u_2,u_3\} \right\rbrace $ as shown in~\cite[Fig.~$2$]{DBLP:conf/spaa/DaiF023}, where a kernel graph $\mathcal{G}$ is a subgraph of $G$, s.t., for any two failures $F\subseteq E$, if $s-t$ is connected in $G\setminus F$ then $s-t$ is also connected in $\mathcal{G}\setminus F$.  By~\cite[Definition~$6.2$]{DBLP:conf/spaa/DaiF023},  a forwarding scheme $\Pi^{(s,t)}$ defines a link-circular forwarding function at each node of $\mathcal{G}$, and we can easily verify that $\Pi^{(s,t)}$ is $2$-resilient against static failures.  In this figure,  $\Pi^{(s,t)}$ is illustrated by solid (red) arcs, dotted (green) arcs, and dashed (blue) arcs respectively, s.t., at a node $v$, a packet from an incoming arc $(u,v)$ is forwarded to an outgoing arc $(v,w)$ that has the same dash pattern (color) as $(u,v)$. If an outgoing arc $(v,w)$ is failed, then the arc $(w,v)$ is treated as an incoming arc to continue forwarding on the dash pattern (color) of $(w,v)$, while a packet originated at $s$ can select  either the solid (red) arc $(s,v_{10})$ or the dashed (blue) arc $(s,u_{10})$ arbitrarily to start.  However, this forwarding scheme $\Pi^{(s,t)}$ is not $2$-resilient against semi-dynamic failures. For semi-dynamic failures $F=\left\lbrace \{v_1,v_2\}, \{v_7,v_9\} \right\rbrace $, by starting at $s$ and following forwarding rules (red arcs), the packet goes through $\left(s, v_{10}, v_0, v_5, v_1, v_2, v_7 \right) $ to meet the first failure $\left(v_7, v_9 \right) $, and then it is rerouted by the dashed forwarding rules (green arcs) to traverse $\left(v_7, v_2 \right) $ to hit the second failure $\left(v_2,v_1 \right) $.  Now, $\Pi^{(s,t)}$ makes the packet stuck in the connected component on $\{v_2, v_7\}$, but in the graph $G\setminus F$, there is still a path from $v_7$ to $t$, e.g., $\left(v_7, v_2, v_3, v_4, v_8, v_9, v_{11}, t \right) $, implying that $\Pi^{(s,t)}$ is not $2$-resilient against semi-dynamic failures. Moreover, after adapting $\Pi^{(s,t)}$ by additionally enforcing clockwise link-circular routing at $v_1$ and $v_2$ to include $\left\lbrace \{v_1,v_4\}, \{v_2,v_3\} \right\rbrace $, we can easily verify that it becomes a $2$-resilient source-matched routing against semi-dynamic failures.}\label{fig:fig_SPAA}
\end{figure}
\begin{theorem}
	For a general graph $G$,  there exists a $1$\nobreakdash-resilient routing scheme against dynamic failures. \label{thm: 1-resilience-always}
\end{theorem}
\begin{proof}
	If a general graph $G=\left( V,E\right) $ is $2$-connected, then Theorem~\ref{thm: 3-resilience} directly implies that $G$ admits a $1$\nobreakdash-resilient routing scheme against dynamic failures. Furthermore, if $G$ is $1$-connected, let $E'\subset E$ denote a set of \emph{bridges}, s.t., $\forall e\in E',$ $G\setminus \{e\}$ is disconnected.  Clearly, each connected component $H_i$ in $G\setminus E'$ must be $2$-connected. Logically, by allowing two parallel edges for each $\{u,v\}\in E'$, we can obtain a $2$-connected (logical) graph $G'=G\cup E'$. We can compute two arc-disjoint arborescences $\{T_0, T_1\}$ of $G'$, s.t., $\forall \{u,v\}\in E'$, either $\left( u,v\right) $ or $\left(v, u\right) $ is included in both $T_0$ and $T_1$.
	  For a dynamic failure $e\in E$, routing along $T_0$ either reaches $t$ or hit $e$ at a node $u\in V$, and if $e\notin E'$, rerouting through the directed path $P_{u,t}$ of $T_1$ will reach $t$ without hitting $e$ anymore. However, if $e\in E'$, the packet cannot arrive at $t$ anymore since the destination $t$ is  in the different connected component in $G\setminus \{e\}$.  
	  We also note that if routing passes through an arc \((u,v) \in T_0\) (or \((u,v) \in T_1\)), which satisfies \(\{u,v\} \in E'\), then the routing always switches to the directed path in \(T_0\) starting from \(v\).
\end{proof}

Chiesa et al.~\cite{DBLP:journals/ton/ChiesaNMGMSS17} shows that $2$-connected graphs cannot admit $2$-resilience against static failures. Conversely, Dai et al.~\cite{DBLP:conf/spaa/DaiF023}, develop a $2$-resilient routing algorithm against static failures by additionally matching the source. In Theorem~\ref{thm: counter-Dai-SPAA'23}, we will first demonstrate that  the  $2$\nobreakdash-resilient source-matched routing algorithm proposed by Dai et al.~\cite[Algorithm~$1$]{DBLP:conf/spaa/DaiF023}  cannot work for two semi-dynamic failures anymore.

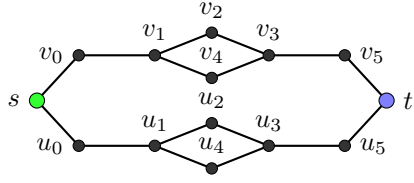
\begin{figure}[t]
	\centering
	\begin{tikzpicture}


	\tikzset{mynode/.style={circle,draw,fill=black!80, outer sep=0pt,inner sep=1.5pt}};
	
	\tikzset{source/.style={circle,draw,fill=green!80, outer sep=0pt,inner sep=2pt}};
	
	\tikzset{destination/.style={circle,draw,fill=blue!50, outer sep=0pt,inner sep=2pt}};
	
	
	\begin{scope}
	\draw[] node[mynode,  label=above:$v_1$]  (v1) at (-0.75,0)  {};

\draw[] node[mynode,  label=above:$v_3$]  (v3) at (0.75,0)  {};

\draw[] node[mynode,  label=above:$v_2$]  (v2) at (0,0.3)  {};

\draw[] node[mynode,  label=above:$v_4$]  (v4) at (0,-0.3)  {};

\draw[] node[mynode,  label=right:$v_5$]  (v5) at (1.75,0)  {};
\draw[] node[mynode,  label=left:$v_0$]  (v0) at (-1.75,0)  {};

\draw[] node[destination,  label=right:$t$]  (t) at (2.3,-0.6)  {};

\draw[] node[source,  label=left:$s$]  (s) at (-2.3,-0.6)  {};

\draw[thick] (v1) to[] (v2);
\draw[thick] (v2) to[] (v3);
\draw[thick] (v3) to[] (v4);
\draw[thick] (v4) to[] (v1);
\draw[thick] (v3) to[] (v5);
\draw[thick] (v0) to[] (v1);

\draw[thick] (s) to[] (v0);
\draw[thick] (t) to[] (v5);
\end{scope}

\begin{scope}[shift={(0,-1.2)}]
	\draw[] node[mynode,  label=above:$u_1$]  (v1) at (-0.75,0)  {};

\draw[] node[mynode,  label=above:$u_3$]  (v3) at (0.75,0)  {};

\draw[] node[mynode,  label=above:$u_2$]  (v2) at (0,0.3)  {};

\draw[] node[mynode,  label=above:$u_4$]  (v4) at (0,-0.3)  {};

\draw[] node[mynode,  label=right:$u_5$]  (v5) at (1.75,0)  {};
\draw[] node[mynode,  label=left:$u_0$]  (v0) at (-1.75,0)  {};

\draw[thick] (v1) to[] (v2);
\draw[thick] (v2) to[] (v3);
\draw[thick] (v3) to[] (v4);
\draw[thick] (v4) to[] (v1);
\draw[thick] (v3) to[] (v5);
\draw[thick] (v0) to[] (v1);

\draw[thick] (s) to[] (v0);
\draw[thick] (t) to[] (v5);

\end{scope}

\end{tikzpicture}
	\caption{Counter-example topology $G$ for $2$-resilient source-matched routing scheme against dynamic failures, where $s$ is the source and $t$ is the destination. Let $V'=\left\lbrace v_0,\ldots, v_5 \right\rbrace $ and $U'=\left\lbrace u_0, \ldots, u_5 \right\rbrace $. By symmetry, w.l.o.g., we can assume $\pi_{s}\left(\bot \right) =v_0$ when $F_{s}=\emptyset$. Then, we can show that each node $v\in V'\cup \{s\}$ must use a link-circular routing function, which has only two possible orderings for its neighbors, i.e., clockwise and counter-clockwise for the shown drawing. For instance, the clockwise and counter-clockwise orderings for $v_1$ are $\left<v_{0},v_2,v_4\right>$ and  $\left<v_0,v_4,v_2\right>$, respectively.  We can further show that $v_1$ and $v_3$ must have the same type of orderings (clockwise or counter-clockwise), otherwise a routing loop can occur, e.g., if $v_1$ and $v_3$ select clockwise and counter-clockwise orderings respective, then a loop $\left(s,v_{0},v_1,v_2,v_3,v_4, v_1,v_0,s\right) $ occurs for a static failure $F=\left\lbrace s,u_0 \right\rbrace $. When $v_1$ and $v_3$ both use the clockwise (resp., counter-clockwise) ordering, for a dynamic failure $\{v_2, v_3\}\in F$ (resp., $\{v_3, v_4\}\in F$) , let $\left(v_2, v_3 \right) $ (resp., $\left(v_4, v_3 \right) $) be down and $\left(v_3, v_2 \right) $ (resp., $\left(v_3, v_4 \right)$)  be up. Then,  a routing loop: $\left( s,v_{0},v_1,v_2,v_1,v_4, v_3,v_2,v_1\right) $ (resp., $\left( s,v_{0},v_1,v_4,v_1,v_2, v_3,v_4,v_1\right)$) appears and the packet originated at $s$ cannot  reach $t$ even there is an $s-t$ a path containing no dynamic failure. A similar proof can be given when $\pi_{s}\left(\bot \right) =u_0$ for $F_{s}=\emptyset$.}
	 \label{fig:fig3}
\end{figure}
\begin{theorem}
There exists a general graph $G$, where the  $2$\nobreakdash-resilient source-matched routing algorithm proposed by Dai et al.~\cite[Algorithm~$1$]{DBLP:conf/spaa/DaiF023} for static failures cannot work for two semi-dynamic failures in $G$ even if $G$ admits a  $2$\nobreakdash-resilient source-matched routing against semi-dynamic failures. \label{thm: counter-Dai-SPAA'23}
\end{theorem}
\begin{proof}
We give a counter-example in Fig.~\ref{fig:fig_SPAA}, where Theorem~\ref{thm: counter-Dai-SPAA'23} can be applied. The main ideas of the proof are explained in the caption of Fig.~\ref{fig:fig_SPAA}.
\end{proof}
Next, by Theorem~\ref{thm: 2-resilient-with-1bit}, we reveal that no $2$\nobreakdash-resilient source-matched routing scheme  can tolerate two dynamic failures and we illustrate the proof ideas of  Theorem~\ref{thm: 2-resilient-with-1bit} in Fig.~\ref{fig:fig3}.
\begin{theorem}\label{thm: 2-resilient-with-1bit}
There exists a $2$\nobreakdash-edge-connected (planar) graph $G$ as shown in  Fig.~\ref{fig:fig3}, where  it is impossible to have a $2$\nobreakdash-resilient source-matched routing scheme against dynamic failures without rewriting bits in packet headers.
\end{theorem}
\begin{proof}
	We first assume that a $2$-resilient source-matched forwarding scheme $\Pi^{(s,t)}$  exists in the graph $G $ as shown in Fig.~\ref{fig:fig3}.
Then, for contradiction, we show that a packet  originated at $s$ cannot be routed to the destination $t$ anymore, but is forwarded in a loop  when there are   two dynamic failures $F$  in $G$, even if there exists an $s-t$ path  in the graph $G\setminus F$.

Let $V'=\left\lbrace v_0,\ldots, v_5 \right\rbrace $ and $U'=\left\lbrace u_0, \ldots, u_5 \right\rbrace $. For each  node $v\in V\left( G\right) $, we define a forwarding function $\pi_{v}\left(u,F_v \right) $ at  $v$, where $F_{v}\subseteq F$ denotes a subset of dynamic failures $F$ that incidents on  $v$ and the source-destination $(s,t)$ is used implicitly in this proof. Clearly, the induced graphs $G[U']$ and $G[V']$ are symmetric. By symmetry,  when $F_{s}=\emptyset$, an arbitrary node in $\{v_{0}, u_{0}\}$ can be chosen as the outgoing port for the packet originated at $s$. W.l.o.g., we  assume that $v_{0}$ is chosen, i.e., $\pi_{s}\left(\bot \right) =v_0$ for $F_{s}=\emptyset$.

Let dynamic failures $F\subseteq E\left( G\right) $ be a set of arbitrary links, s.t.,  $\left|F \right|\leq 2 $ and $F$ can be empty. Next, we claim that, given $\pi_{s}\left(\bot \right) =v_0$ with $F_s=\emptyset$,  for each node $v\in V'\cup \{s\} $, its routing function must be \emph{link-circular} even when $F$ are static failures. If $v\in  V'\cup \{s\}$ has  $\Delta_{G\setminus F}\left( v\right)=1$, then $\pi_{v}\left(u,  F_v\right) =u$, where $\pi_{v}\in\Pi^{(s,t)} $ and $u\in E_{G\setminus F}\left( v\right) $ denotes its unique neighbor in $G\setminus F$, otherwise packets get stuck at $v$. This case can be thought as a special case of the link-circular forwarding.  If $v\in  V'\cup \{s\} $ has  $\Delta_{G\setminus F}\left( v\right)=3$, i.e., $\Delta_{G}\left( v\right)=\Delta_{G\setminus F}\left( v\right)$ and $F_v=\emptyset$, a \emph{non-link-circular} forwarding function at $v$ must imply $\exists x,y\in N_{G\setminus F}\left(v\right): \pi_{v}\left(x\right) =y$ and  $\pi_{v}\left(y  \right) =x$, where $N_{G\setminus F}\left( v\right)  =\left\lbrace x, y, z\right\rbrace $ are neighbors of $v$ in $G\setminus F$. However, a non-link-circular forwarding function cannot be $2$-resilient if the only  $s-t$ path remained in $G\setminus F$ has to go through the link $\{v, z\}$.
For example,  when $F=\left\lbrace \{s,u_{0}\},\{v_2,v_3\}\right\rbrace $ and $\Delta_{G\setminus F}\left(v_{1}\right)=3$, if a non-link-circular forwarding function  has $\pi_{v_1}\left(v_{0}  \right) =v_2$ and  $\pi_{v_1}\left(v_{2} \right) =v_{0}$,  then a packet starting at $s$ cannot approach $t$ anymore even if $s-t$ is connected via $\{v_1,v_4\}$. A similar argument can be established if $\pi_{v_1}\left(v_{0} \right) =v_4$ and  $\pi_{v_1}\left(v_{4}\right) =v_{0}$, and  $F=\left\lbrace \{s,u_{0}\},\{v_4,v_3\}\right\rbrace $. Moreover, for each $v\in  V'\cup \{s\}$ having  $\Delta_{G\setminus F}\left( v\right)=2$, a non-link-circular forwarding function   at $v$ must imply $\exists x\in N_{G\setminus F}\left(v\right): \pi_{v}\left(x\right) =x$ for $N_{G\setminus F}\left( v\right)  =\{x, y\}$, which can make $v$ become a dead-end node, i.e., a packet cannot traverse from one neighbor of $v$ to the other neighbor to approach $t$ anymore. Therefore, each $v\in  V'\cup \{s\} $ must have a link-circular forwarding function.

If $v\in  V'\cup \{s\} $ has  $\Delta_{G\setminus F}\left( v\right)=2$, then its link-circular forwarding function is unique, i.e., from one neighbor to the other neighbor.  If $v\in V'\cup \{s\} $ has  $\Delta_{G\setminus F}\left( v\right)=3$, where $F_v=\emptyset$, then there are two possible circular orderings for its neighbors $N_{G\setminus F}\left( v\right) $, i.e., one clockwise and the other counter-clockwise based on their geometric locations in Fig.~\ref{fig:fig3}. For example, at $v_1$, the clockwise ordering of $N_{G}\left( v_1\right) $ is $\left<v_{0}, v_2,v_4\right>$ and the counter-clockwise ordering of $N_{G}\left( v_0\right)$ is $\left<v_{0}, v_4, v_2\right>$. Thus, for each $v\in  V'\cup \{s\} $ that has  $\Delta_{G\setminus F}\left( v\right)=3$, its link-circular forwarding function can choose one of two options: clockwise and counter-clockwise, arbitrarily.

Fixing $\{s,u_{0}\}\in F$, let $\{v_2,v_3\}\in F$ (resp., $\{v_4,v_3\}\in F$) be a dynamic failure in $G\left[ V'\cup \{s,t\}\right] $ if $v_0$ and $v_9$ both have the clockwise (resp., counter-clockwise) of link-circular forwarding functions.  In this case, even if $s,t$ are connected in $G\left[ V'\cup \{s,t\}\right] $, a packet originated at $s$ will enter a forwarding loop: $\left(v_0, v_1, v_2, v_1, v_4, v_3, v_2, v_1\right) $ (resp., $\left(v_0, v_1, v_4, v_1, v_2, v_3, v_4, v_1\right) $), where the dynamic failure $\{v_2,v_3\}$ (resp., $\{v_4,v_3\}$) is down only for the first hitting but always up since then,  but never traverses the link $\{v_3, v_{5}\}$ to arrive at $t$. When  $v_1$ and $v_3$ have the different type, by fixing $\{s,u_{0}\}\in F$, even if there is no dynamic failure in $G\left[ V'\cup \{s,t\}\right] $, a forwarding loop: $(s, v_0, v_1, v_2, v_3, v_4, v_1, v_0, s)$ occurs if $v_1$ and $v_3$ take  forwarding functions of clockwise and counter-clockwise orderings respectively, otherwise another forwarding loop: $(s, v_0, v_1, v_4, v_3, v_2, v_1, v_0, s)$ exists. Moreover, a similar discussion can be applied  when $\pi_{s}\left(\bot \right) =u_0$ for $F_{s}=\emptyset$.

Thus, no  $2$-resilient source-matched  forwarding scheme against dynamic failures for $(s,t)$ exists  in~Fig.~\ref{fig:fig3}.
\end{proof}
For the counter-example graph $G$  as shown in~Fig.~\ref{fig:fig3}, fixing $\pi_{s}\left(\bot \right) =v_0$ when $F_{s}=\emptyset$, the routing functions at  $v_1$ and $v_3$ cannot know whether an incoming packet currently should either continue searching  a path towards $t$ in $G\left[V'\cup\{s,t\}\right]$ or finding a path back to $s$ in $G\left[V'\cup\{s,t\}\right]$ to try paths in $G\left[U'\cup\{s,t\}\right]$. Simply, by rewriting one bit in packet headers, the source-matched routing functions can resolve this weakness to achieve $2$-resilience against dynamic failures again in $G$. Now, a fundamental question arises: Can we achieve  $k$-resilience against dynamic failures in a general graph by only modifying $O\left( \log k\right) $ bits? 

In light of Theorem~\ref{thm: no-perfect-resilience}, we demonstrate that achieving perfect resilience through the modification of $O\left( \log k\right) $ bits is impossible.
\begin{figure}[t]
	\centering
	\begin{tikzpicture}


	\tikzset{mynode/.style={circle,draw,fill=black!80, outer sep=0pt,inner sep=1.5pt}};
	
	\tikzset{source/.style={circle,draw,fill=green!80, outer sep=0pt,inner sep=2pt}};
	
	\tikzset{destination/.style={circle,draw,fill=blue!50, outer sep=0pt,inner sep=2pt}};
	

%

  \draw[very thick] (-0.8,2.) node[mynode,label=below:\large$v_{i}$]{} -- (0.8,2.) node[mynode,label=below:\large$v_{j}$]{};

\draw[blue, fill=green, -{Triangle[width = 18pt, length = 8pt]}, line width = 10pt] (0.0, 1.8) -- (0.0,1.2);

%

%
%


\node[cloud,
draw =blue,
text=black,
fill = gray!10,
aspect=1.5,
cloud puffs = 16, thick] (c) at (0,0) {\Large$H^{i,j}$};

\draw[] node[mynode,  label={[shift={(0.2,-0.1)}]:\large$u^{i,j}_t$}]  (v5) at (0.88,0)  {};
\draw[] node[mynode,  label={[shift={(-0.2,-0.1)}]:\large$u^{i,j}_s$}]  (v0) at (-0.88,0)  {};


\draw[dashed, thick] (-2.3,-0.9) rectangle node{} (2.3,0.9);

\draw[] node[mynode,  label=above:\large $v_{j}$]  (vj) at (1.8,0)  {};
\draw[] node[mynode,  label=above:\large$v_{i}$]  (vi) at (-1.8,0)  {};

\draw[very thick] (vj) to[] (v5);
\draw[very thick] (vi) to[] (v0);

	
	\draw[dashed, thick] (-2.3,-0.9) rectangle node{$G_{i,j}$} (-1,-0.3);

\draw[] node[]  () at (1.8,0)  {$$};

\end{tikzpicture}
	\caption{An illustration of constructing the graph $G'$ in the proof of Theorem~\ref{thm: no-perfect-resilience}, where  we replace each edge $\{v_i,v_j\}$ in the graph $G$ as shown in Fig.~\ref{fig:fig3} with another graph $G_{i,j}= H^{i,j}\cup \left\lbrace \{v_i, u_s^{i,j}\}, \{u_t^{i,j},v_j\} \right\rbrace$.}\label{fig: non-constant-bits}
\end{figure}
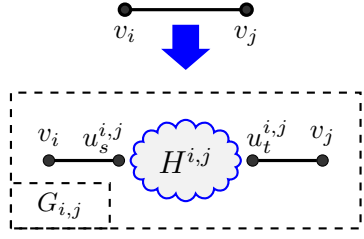
\begin{theorem}
	There exists graphs for which any resilient source-matched routing that can tolerate $2k$ dynamic failures needs rewriting of at least $k$ bits in packet headers for $k\in \mathbb{N}$. \label{thm: no-perfect-resilience}
\end{theorem}
\begin{proof}
We prove Theorem~\ref{thm: no-perfect-resilience} by induction on $k$.
Theorem~\ref{thm: 2-resilient-with-1bit} implies that  there is a graph $G=\left( V,E\right)$   as shown in~Fig.~\ref{fig:fig3}, where any $2$\nobreakdash-resilient source-matched routing for $\left( v_s, v_t\right) $ needs rewriting at least one bit, proving the initial case of $k=1$.   We assume that there is a general graph $H=\left(U,E_U \right) $, where a $2k$\nobreakdash-resilient source-matched routing  against dynamic failures    for a source-destination pair $\left( u_s, u_t\right) $ (resp., $\left( u_t, u_s\right) $) with $u_s, u_t\in U$  needs rewriting at least $k\ge 2$ bits.

Now, as illustrated in Fig.~\ref{fig: non-constant-bits}, we can construct another graph $G'=\left(V', E' \right) $ by replacing each edge $\{v_i, v_j\}\in E$ in $G$  with a graph $G_{i,j}= H^{i,j}\cup \left\lbrace \{v_i, u_s^{i,j}\}, \{u_t^{i,j},v_j\} \right\rbrace$, where $H^{i,j}=\left( U^{i,j}, E^{i,j}_U\right) $ is isomorphic to $H$, i.e.,  $U^{i,j}=\left\lbrace u^{i,j}_\ell:  u_\ell \in U\right\rbrace $ and  $E_U^{i,j}=\left\lbrace \{u^{i,j}_\ell, u^{i,j}_o\}:  \{u_\ell, u_o\} \in E_U\right\rbrace$,
and nodes $ u_s^{i,j}, u_t^{i,j}\in U^{i,j}$.  We note that we use $v_s$ and $s$ (resp., $v_t$ and $t$) interchangeably in this proof.

Next, we claim that  any $\left( 2k+2\right) $\nobreakdash-resilient source-matched routing for $\left( v_s, v_t\right) $ in $G'$  needs rewriting at least $k+1$ bits. Let $F'=F^{i,j}\cup F$ be any $2k+2$ dynamic failures in $E'$, where $F^{i,j}\subset E_U^{i,j}$ for $\{v_i, v_j\}\in E$ has $\left| F^{i,j}\right| =2k$ and $F\subseteq \left\lbrace \left\lbrace  \{v_i, u_s^{i,j}\}, \{u_t^{i,j},v_j\}\right\rbrace : \{v_i, v_j\}\in E \right\rbrace $  has $\left| F\right| =2$.

In the graph $G$ as shown in Fig.~\ref{fig:fig3}, we always assume $\pi_{s}\left(\bot \right) =v_0$ for  $F_{s}=\emptyset$.
Let $F^*=\{e_1, e_2\}$ denote two dynamic failures in $G$. For example,  when $e_1=\{v_5, t\}$ and $e_2=\{v_1, v_4\}$,  the packet starting at $s$  meets the first failure $\left( v_5,t\right) $, and it has to go through $\left(v_2, v_1 \right) $ back to $v_1$ since $\{v_1, v_4\}$ failed.  Clearly,   one bit in the packet header must be rewritten to inform $v_1$ whether the packet has already visited $v_3$, s.t.,  $v_1$ can decide forwarding it  to $s$ or $v_4$ provided that $\{v_1, v_4\}$ is recovered. Similarly, for $F^*$ in $G'$, we set $\{u_t^{5,t}, v_t\}, \{u_t^{1,4},v_4\}\in F$ and $F^{i,j}=F^{1,2}$, we still need one bit at $v_1$ to indicate whether the packet has already visited $v_3$. Still, to go through $H^{1,2}$ to arrive at $v_1$, we need rewriting of additional $k$ bits under failures $F^{i,j}$.
Thus, we need rewriting $k+1$ bits for $2k+2$ dynamic failures in $G'$. For other failure cases in $G$, we can similarly map them  to scenarios in~$G'$.
\end{proof}

\section{Conclusions and Future Work}\label{sec:conclusion}
This paper explored the achievable resilience and limitations of failover routing mechanisms in the presence of static, semi-dynamic and dynamic failures. Our results demonstrate that achieving the ideal resilience, i.e., the $(k-1)$-resilience in $k$-edge-connected graphs, for $k\leq 5$ is possible for dynamic failures and can be extended to any $k$ by rewriting $\log k$ bits in packet headers. However, we find out that the previously proposed 3-bits header-rewriting algorithm by Chiesa et al.~\cite{DBLP:journals/ton/ChiesaNMGMSS17} falls short of achieving the ideal resilience  for dynamic failures, although it remains effective for semi-dynamic failures.  Pessimistically, our theorems on general graphs indicate that only 1-resilience is attainable without bit rewriting, and achieving arbitrary $k$-resilience against dynamic failures becomes impossible even with the ability to rewrite $\log k$ bits.

Our work leaves open several interesting avenues for future research, particularly in exploring the ideal $k$-resilience of an arbitrary $k$ by only rewriting $O\left( 1\right) $ bits in more specific dynamic scenarios, both analytically and empirically.

\bibliography{literature_clean_short}

\appendix

\section{First Insights for Ideal Resilience against Static Failures}\label{sec: techniques}
In this section, we initially present the routing techniques proposed by Chiesa et al.~\cite{DBLP:journals/ton/ChiesaNMGMSS17} to achieve $\left(k-1 \right) $-resilience against static failures in $k$-connected graphs $G$. We will demonstrate that the results established for static failures can be effectively adapted for dynamic failures.

Chiesa et al.~\cite{DBLP:journals/ton/ChiesaNMGMSS17} leverage  \emph{a set of $k$ arc-disjoint arborescences}~\cite{TARJAN197451}, in a $k$\nobreakdash-connected graph $G$ to devise their resilient failover protocols.

\subparagraph*{Arc-Disjoint Arborescences.}
An \emph{arborescence $T$} of a graph $G=\left(V,E \right) $ is a \emph{directed spanning tree} of $G$,  rooted at a node $t\in V$, s.t.,   each node $v\in V\setminus \{t\}$ has a unique directed path  from $v$ to $t$ on $T$. A set of arborescences $\mathcal{T}=\left\lbrace T_1, \ldots, T_k \right\rbrace $ of $G$ is \emph{arc-disjoint} (resp., \emph{edge-disjoint}) if  two arbitrary arborescences  $T_i\in \mathcal{T}$ and $T_j\in \mathcal{T}\setminus T_i$ do not share any arc (resp., any edge after removing directions of arcs on $T_i$ and $T_j$). We note that two arc-disjoint arborescences can share common edges. We can compute  $k$ arc-disjoint arborescences in a $k$-edge-connected graph efficiently~\cite{TARJAN197451}, both in theory (in $O\left( \left| E\right| k\log n + nk^4\log^2 n\right)$~\cite{10.5555/1347082.1347132}) and in practice~\cite{DBLP:conf/wea/GeorgiadisKMN22}.

\begin{lemma}[\hspace{1sp}{\cite[Lemmas~$4$ and $5$]{DBLP:journals/ton/ChiesaNMGMSS17}}]
	For any $2k$\nobreakdash-connected (resp., $\left( 2k+1\right) $\nobreakdash-connected) graph $G$, with $k\ge 1$, and a node $t\in V$, there exist $2k$ (resp., $2k+1$) \emph{arc-disjoint arborescences} $T_1, \ldots, T_{2k}$ (resp., $T_1, \ldots, T_{2k+1}$) rooted at $t$ such that $T_1, \ldots, T_{k}$ do not share edges with each other and $T_{k+1}, \ldots, T_{2k}$ do not share edges with each other. \label{lem: arc-disjoint-arborescences}
\end{lemma}
In Lemma~\ref{lem: arc-disjoint-arborescences}, the set of \emph{edge-disjoint} arborescences $T_1, \ldots, T_{k}$ (resp., $T_{k+1}, \ldots, T_{2k}$) will be called \emph{left arborescences} (resp., \emph{right arborescences}). For a set of $2k+1$ arc-disjoint arborescences $\mathcal{T}=\left\lbrace T_1, \ldots, T_{2k+1}\right\rbrace $ in a $\left( 2k+1\right) $\nobreakdash-connected graph $G$, there exist an arborescence  $T_{2k+1}\in \mathcal{T}$ that may share edges with left and right arborescences simultaneously.

\subparagraph*{Arborescences-Based Routing.}
Without matching the source, Chiesa et al.~\cite{DBLP:journals/ton/ChiesaNMGMSS17} developed a series of \emph{arborescences-based routing}  functions to achieve ideal resilience against $k-1$ \emph{static} failures in a $k$-connected graph of $k\ge 2$.  Hereinafter, unless specified otherwise, we use  $\mathcal{T}=\left\lbrace T_1, \ldots, T_{k}\right\rbrace $ to denote a set of $k$ arc-disjoint arborescences rooted at the same node $t$ in a $k$\nobreakdash-connected graph $G$.

Next, we will first present a number of elemental routing modes based on a set of $k$ arc-disjoint arborescences $\mathcal{T}$, which were introduced by  Chiesa et al.~\cite{DBLP:journals/ton/ChiesaNMGMSS17} to devise their more sophisticated routing functions, e.g., header-rewriting. 
For an arbitrary arborescence $T_i\in \mathcal{T}$, in a \emph{canonical mode}, a packet at a node $v\in V$ is routed along the unique $v-t$ path  defined on $T_i$~\cite{DBLP:journals/ton/ChiesaNMGMSS17}. 
If a packet traversing along $T_i\in \mathcal{T}$ in canonical mode hits a failure (arc) $\left(u,v\right) $ at a node $u$, where $(u,v)\in E\left(  T_i\right) $ and $\{u,v\}\in E$, Chiesa et al.~\cite{DBLP:journals/ton/ChiesaNMGMSS17} introduce two possible routing actions:
\begin{itemize}
	\item \emph{Next available arborescence}: After seeing the failed arc  $\left(u,v\right) $ on $T_i$, a packet will be rerouted on the next available arborescence $T_{\textnormal{next}}\in \mathcal{T}$ on a predefined ordering of arborescences in $\mathcal{T}$ starting at $u\in V$, i.e., $T_{\textnormal{next}}= T_{\left( j \mod k\right) }$, where $j\in \{i+1, \ldots, i+k\}$ is the minimum number, s.t., there is no  failed arc on $T_{\left( j \mod k\right) }$ starting at $u$. 
	
	\item \emph{Bounce back on the reversed arborescence:} a packet hitting a failure $\left(u,v\right) $ on $T_i$ at the node $u$ will be rerouted along the arborescence $T_{j}\in \mathcal{T}$ that contains the arc $(v,u)$, i.e., $(v,u)\in E\left(T_j\right) $, starting at $u\in V$.
\end{itemize}

\begin{definition}[Circular-Arborescence Routing~\cite{DBLP:journals/ton/ChiesaNMGMSS17}]\label{def: circular-arborescence}
	Given   a set of $k$ arc-disjoint arborescences $\mathcal{T}=\left\lbrace T_1, \ldots, T_{k}\right\rbrace $ of a graph $G$, a circular-arborescence routing defines a circular-ordering $\left\langle \mathcal{T}\right\rangle $ of $\mathcal{T}$, and for a packet originated at $v\in V$, it selects an arbitrary $T_i\in \left\langle \mathcal{T}\right\rangle$ (usually $T_i$ is the first one) to send the packet along $T_i$ from $v$ in canonical mode and when hitting a failure at a node $v_i$, it reroutes along the next available arborescence $T_j\in \mathcal{T}$ of $T_i$ based on   $\left\langle \mathcal{T}\right\rangle $ from $v_i$ in canonical mode, and so on if more failures are met until arriving at the destination.
\end{definition}

Chiesa et al.~\cite{DBLP:journals/ton/ChiesaNMGMSS17} show that there must exist a circular-arborescence routing, which is $\left(k-1 \right) $\nobreakdash-resilient against static failures in a $k$-connected graph with $k=2,3,4$. However, since the effectiveness of circular-arborescence routing is unapparent for $k\ge 6$, Chiesa et al.~\cite{DBLP:journals/ton/ChiesaNMGMSS17}  introduced a novel algorithmic toolkit, \emph{meta-graph}, to delve deeper into the connection between a fixed set of $k-1$ failures $F$ and arborescences of $\mathcal{T}$, which further enhances understanding of routing behaviors resulting from bouncing back on the reversed arborescences. 
It is worth mentioning that constructing the meta-graph is not required for computing routing functions; however, it serves as a helpful aid in constructing proofs.

\subparagraph*{Meta-graph.}
Given $k$ arc-disjoint arborescences $\mathcal{T}=\left\lbrace T_1, \ldots, T_{k}\right\rbrace $ of $G=\left(V,E \right) $, for a set of static failures $F\subset E$, where $\left| F\right| =f \leq  \left( k-1\right) $, Chiesa et al.~\cite{DBLP:journals/ton/ChiesaNMGMSS17} define a meta-graph $H_F=\left( V_F, E_F\right) $ as follows: each node $\mu_i\in V_F$, where $i\in \{1,\ldots, k\}$, represents an arborescence $T_i\in \mathcal{T}$; and for each failure $\{u,v\}\in F$, if $\left(u,v \right)\in E\left( T_i\right)  $ and $\left(v,u \right)\in E\left( T_j\right)  $, then there is an edge $\{\mu_i, \mu_j\}\in E_F$, and if either $\left(u,v \right)\in E\left( T_i\right)  $ or $\left(v,u \right)\in E\left( T_j\right)  $, then there is a self-loop edge at either $\mu_i$ or $\mu_j$ in $E_F$. We also note that $H_F$ might contain parallel edges and multiple connected components. Chiesa et al.~\cite[Lemma~1]{icalp16} shows that, for any  $F\subseteq E$ of $\left| F\right| \leq k-1$ static failures in a $k$-connected graph $G$, the corresponding meta-graph $H_F$ must contain at least $k-f$ trees. For dynamic failures $F$, we define $H_F$ as the \emph{maximum meta-graph} for dynamic failures $F$ by assuming that links in $F$ permanently and simultaneously fail. Then, at any time point, since dynamic failures in $F$ can be up or down arbitrarily, the \emph{real-time meta-graph} $H'_F\subseteq H_F$ induced by $F$ must be a subgraph of $H_F$, where an edge in $H_F$ can also occur in $H'_F$ arbitrarily. In the following,  
	since a subgraph $H'_F$ of $H_F$ does not impact our discussion, we also use $H_F$ to denote a real-time meta-graph implicitly.
By Lemma~\ref{lem: meta-graph}, we will show that $H_F$ contains at least one tree for dynamic failures $F$.

\begin{lemma}
	For a set of dynamic failures $F\subset E $, where $\left| F\right|=f \leq k-1$,  the set of connected components of meta-graph $H_F$ contains at least $k-f$ trees. \label{lem: meta-graph}
\end{lemma}
\begin{proof}
	Chiesa et al.~{\cite[Lemma~1]{icalp16}} gave a proof of Lemma~\ref{lem: meta-graph} for static failures. Recall that each edge in a meta-graph $H_F$ implies a link failure $e\in F$. Given a tree  $h\in H_F$ for static failures $F$, let $h'\subseteq h$ be a subgraph of $h$, where  edges of $h$ might arbitrarily occur in $h'$. Then, there exists a tree, denoted by $h'\subset h$, contained in $H_F$ when failures $F$ become dynamic/semi-dynamic.
\end{proof}

\subparagraph*{Good Arborescences.}
Given $\mathcal{T}=\left\lbrace T_1, \ldots, T_{k}\right\rbrace $ of $G=\left(V,E \right) $ and arbitrary $k-1$  static failures $F\subset E$, an  $T_i\in \mathcal{T}$ is called \emph{a good arboresence} if  from any node $v\in V$, routing a packet along $T_i$ in canonical mode will either reach the destination $t$ uninterruptedly or hit a failed arc $\left( u,v\right)\in E\left(T_i \right)  $ on $T_i$, s.t., bouncing back along the reversed arborescence $T_j\in \mathcal{T}$, where $\left( v,u\right) \in E\left( T_j\right) $, reaches  $t$ directly without hitting any more failure on $T_j$.
By Chiesa et al.~{\cite[Lemma~4]{icalp16}}, there is always a good arborescence, which is represented by a node $v\in V_F$ contained in a tree component in $H_F$, when $F$ are static. We will show that this conclusion can be extended  to dynamic failures $F$ by~Theorem~\ref{thm: good-arborescence}.

\subparagraph*{Well-Bouncing.} If a bouncing  from $T_i$ to $T_j$ on a failure $(u,v)$, where $(u,v)\in E\left(T_i \right) $ and $(v,u)\in E\left( T_j\right) $, will reach $t$ directly along $T_j$ without hitting any failure, then this bouncing is called well-bouncing. Clearly, bouncing on any failure of a good-arborescence is well--bouncing.

\begin{theorem}
	Given a set $\mathcal{T}$ of  $k$ arc-disjoint arborescences    of a $k$-connected graph $G=\left(V,E \right) $, for any set of $k-1$ \emph{dynamic} failures $F\subset E$, $\mathcal{T} $ contains at least one good arborescence. \label{thm: good-arborescence}
\end{theorem}
\begin{proof}
	Chiesa et al.~\cite[Lemma~4]{icalp16} gave a proof of Lemma~\ref{lem: meta-graph} when $F$ are static failures. We will first introduce the proof by Chiesa et al.~\cite[Lemma~4]{icalp16}, and then extend their proof ideas to obtain a new proof for  dynamic failures.
	
	For static failures $F$, by the definition of meta-graph $H_F$,  each edge $\{\mu_i, \mu_j\}\in E_F$ can imply two possible occurrences of bouncing on a failure, i.e.,  one from $T_i$ to $T_j$ and one from $T_j$ to $T_i$ respectively.

	If $T\in \mathcal{T}$ and $T'\in \mathcal{T}$ share a failure $\{u,v\}$ and the arc $(v,u)\in E\left( T'\right) $ is the highest failure, i.e., no  failure on the directed path $u-t$ along $T'$, then a bouncing from $T\in \mathcal{T}$ to $T'\in \mathcal{T}$ on $\left( u,v\right) $  is \emph{well-bouncing} since a packet can arrive at the destination $t$ along $T'$ uninterruptedly after bouncing from $T$ to $T'$. Given a tree $h$ in $H_F$, each node in $h$ represents an arborescence $T_i\in \mathcal{T}$, which further implies that there exists a bouncing  to $T_i$ is well-bouncing since at least one failure on $T_i$ is the highest one.

	Thus, for a tree $h$ with $\left| E\left( h\right) \right|= \left| V\left( h\right) \right|-1$,  it implies that at most  $2\left| E\left( h\right) \right|- \left| V\left( h\right) \right| \leq  \left| V\left( h\right) \right|-2$ bouncing  are not well-bouncing. We note that each node in $h$ indicates a distinct arborescence $T\in \mathcal{T}$. Then, there must be one node  in $h$ representing an arborescence $T$, s.t., every bouncing from $T$ is well-bouncing, implying that $T$ is a good arborescence.
	
	Suppose that  $T$ is a good arborescence in a tree-component $h\subset H_F$ for static failures $F$. Now, when $F$ becomes dynamic, a failure $\left( u,v\right) \in E\left( T\right) $ on $T$ may disappear for a canonical routing along $T$, but once we hit a failure $\left( u,v\right)$ on $T$, which must imply a well-bouncing from $T$ to another  arborescence $T'$, as $\left(v,u \right) \in F$ must be the highest failure on $T'$. Therefore, $T$ is also a good arborescence for dynamic failures $F$.
\end{proof}
\subparagraph*{Dilemma of Good Arborescences.}We first note that  meta-graphs $H_F$  can be dissimilar for different failure sets, $F$. Then, any arborescence $T\in \mathcal{T}$ can become the unique good arborescence for a specific set $F$. Thus, finding the good arborescence needs a circular-arborescence routing  on a fixed order $\left\langle  \mathcal{T}\right\rangle $ of $ \mathcal{T}$, which is independent of $F$, s.t., each arborescence $T$ in $ \mathcal{T}$ can be visited. However, to check  whether  $T\in \mathcal{T}$ is a good arborescence, it needs to bounce from the current arborescence $T$ to an arbitrary arborescence $T'\in \mathcal{T}\setminus \{T\}$ when a canonical routing along $T$ hits a failure $e\in F$ that is also shared by $T'$, which means  leaving the fixed order $\left\langle  \mathcal{T}\right\rangle $ of $ \mathcal{T}$ but visiting a random  arborescence in $\mathcal{T}$ depending on $F$.

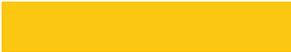
\end{document}